\documentclass[11pt]{amsart}

\usepackage{amssymb,amsfonts}
\usepackage[all,arc]{xy}
\usepackage{enumerate}
\usepackage{mathrsfs}
\usepackage{amsmath}
\usepackage{float}
\usepackage{tikz}
\usetikzlibrary{arrows.meta}
\usepackage{appendix}
\usepackage{geometry}

\usepackage[utf8]{inputenc}
\usepackage{subcaption}
\usepackage{tabulary}
\newtheorem{thm}{Theorem}[section]
\newtheorem{cor}[thm]{Corollary}

\newtheorem{lem}[thm]{Lemma}

\newtheorem{inv}{Invariant}

\theoremstyle{definition}

\theoremstyle{remark}

\makeatletter
\let\c@equation\c@thm
\makeatother
\numberwithin{equation}{section}

\title{A Fast Max Flow Algorithm}

\author{J\lowercase{ames} B. O\lowercase{rlin}\\ MIT, \lowercase{jorlin@mit.edu;}\\X\lowercase{iao}-Y\lowercase{ue} G\lowercase{ong}\\ MIT, 
		\lowercase{xygong@mit.edu}}
\date{July 5, 2018.  Most recent revision:  October 7, 2019.}

\begin{document}

\begin{abstract}
In 2013, Orlin proved that the max flow problem could be solved in $O(nm)$ time. His algorithm ran in $O(nm + m^{1.94})$ time, which was the fastest for graphs with fewer than $n^{1.06}$ arcs.   If the graph was not sufficiently sparse, the fastest running time was an algorithm due to King, Rao, and Tarjan. We describe a new variant of the excess scaling algorithm for the max flow problem whose running time strictly dominates the running time of the algorithm by King et al.  Moreover, for graphs in which $m = O(n \log n)$, the running time of our algorithm dominates that of King et al. by a factor of 
$O(\log\log n)$.

\end{abstract}

\maketitle

\section{Introduction}
Network flow problems form an important class of optimization problems within operations research and computer science. Within that class, the max flow problem, has been widely investigated since the seminal research of Ford and Fulkerson \cite{FF56} in the 1950s. 
A discussion of algorithms and applications for the max flow problem can be found in \cite{AMO93}.  
We consider the max flow problem on a directed graph with $n$ nodes, $m$ arcs, and integer valued arc capacities (possibly infinite), in which the largest finite capacity is bounded by $U$. The fastest strongly polynomial time algorithms are due to Orlin \cite{O13} and King et al. \cite{KRT94}. 
The running time for Orlin's algorithm is $O(nm + m^{31/16} \log^2 n)$, which is $O(nm)$ when $m < n^{1.06}$.   
On networks with $O(n)$ arcs, the algorithm can be further improved to run $O(\frac{n^2}{\log n})$ time.   
The running time for the algorithm by King et al. \cite{KRT94} is $O(nm \log_{\beta}  n)$  where $\beta  =2 + \lfloor \frac{m}{n \log n} \rfloor$.   The running time is $O(nm)$ whenever $m = \Omega(n^{1+\epsilon})$ for any fixed $\epsilon > 0$.  
Their algorithm is based upon the randomized algorithm of Cheriyan and Hagerup \cite{CH89}, which runs in $O(nm + n^2 \log^2 n)$ expected time.

Orlin's algorithm is based on using a fast weakly polynomial time algorithm.   The fastest weakly polynomial time algorithms are those of Goldberg and Rao \cite{GR98} with running time $O( \min\{m^{1/2}, n^{2/3}\} m \log \frac{n^2}{m} \log U)$, and  Lee and Sidford \cite{LS14} with running time $\tilde O (mn^{1/2} \log^2 U)$. 

Here we present a new fast max flow algorithm that runs in $O(kn^2 + nm \log_k n)$ time, where the parameter $k$ can be selected by the user.  Assuming that $n \ge 4$ and $m \ge n$ one can optimize the asymptotic running time by choosing  $k = \lceil \frac{\log n}{\log\log n} + \frac{m}{n} \rceil$.   This leads to a running time of 
$$O\left(\frac{nm \log n}{\log\log n + \log \frac{m}{n}}\right).$$ 

For all values of $n$ and $m$, the running time of our algorithm dominates the running time of King et al. \cite{KRT94}.  When $m \le n \log n$, our algorithm is faster than the algorithm of \cite{KRT94} by a factor of $\log\log n$.

Our algorithm is based on the stack-scaling algorithm of Ahuja et al. \cite{AOT89}, while also drawing upon the idea of ``special pushes'' in the paper by \cite{O93}.  The analysis of special pushes in this paper is similar to that in the paper by Orlin \cite{O13}. 

The contributions of this paper may be described as follows.

\begin{itemize}
\item[$1.$] We present a simple variant of the stack-scaling algorithm of \cite{AOT89} in which there are no stacks.  We refer to the revised algorithm as the {\it Large-Medium Excess-Scaling (LMES) Algorithm}.  

\item[$2.$] We give a new (and simpler) proof of Orlin's Contraction Lemma \cite{O13}, which was used to develop an $O(nm)$ max flow algorithm.   
\item[$3.$] In the modified version of the LMES algorithm (called the Enhanced LMES 
Algorithm), we permit node excesses that are slightly negative.   When the negative excess of 
a node $v$ reaches a threshold value, then node $v$ is added to the  ``flow-return forest,'' which is a data structure designed for the Enhanced LMES.   Within $O( \log_k n)$ additional scaling 
phases, flow is sent to node $v$, after which $e(v) \ge 0$.   
\item[$4.$]  Our algorithm achieves its improved running time without relying on the dynamic 
tree data structure of  \cite{ST83}.   It is an open question as to whether  a further 
speed-up is possible with the use of dynamic trees.
\end{itemize}

It is possible that the flow-return forest may be of use in other algorithms for the max flow problem or closely related problems.


\section{Preliminaries}
Let $G=(N, A)$ be a directed graph with node set $N$ and arc set $A$. Let $n=\lvert N\rvert$, 
and $m=\lvert A\rvert$. Each arc $(i,j)\in A$ is associated with a non-negative integer \textit{capacity} $u_{ij}$.  The assumption of finiteness of capacities is without loss of generality.  For networks with no infinite capacity path from source to sink, one can replace any infinite capacity by $nU$, where $U$ is the largest finite capacity of an arc.

There is a unique \textit{source} node $s \in N$ and a unique \textit{sink} node $t \in N$. For convenience of exposition, we will assume that there is at most one arc from $i$ to $j$ 
and that for each arc $(i,j)\in A$, $(j,i)$ is also in $A$, possibly with $u_{ji} = 0$.   We refer to  
$u_{ij} + u_{ji}$ as the \textit{bi-capacity} of arc $(i, j)$.   By symmetry, it is also the bi-capacity of arc $(j, i)$.

For each node $i \in A$, we let $A^+(i)$ denote the set of arcs directed out of node $i$.  We let $A^-(i)$ denote the set of arcs directed into node $i$.  We let $A(i) = A^+(i) \cup A^-(i)$.   We assume that there is a preprocessing step during which the arcs of $A^+(i)$ and $A^-(i)$  are ordered in non-increasing order of their bi-capacity.   The preprocessing step takes $O(m \log n)$ steps and is not a bottleneck operation of our max flow algorithm.  This assumption is of use later in this paper when we discuss the Enhanced LMES Algorithm.




Suppose that $\alpha$, $\beta$ and $\gamma$ are all integers.   When we write 
$\alpha \equiv \beta ( \kern -8 pt \mod  \gamma )$, we mean 
that there is some integer $\tau$ such that $\alpha = \beta + \tau \gamma$.  
When we write $\alpha = \text{Mod} ( \beta ,  \gamma )$, we mean that   
$\alpha$ is the non-negative remainder obtained by dividing $\beta$ by  
$\gamma$.  

All logarithms in the paper are base two unless an explicit base is given.

A \textit{feasible flow} is a function $x: A\rightarrow\mathbf{R}$ that satisfies the flow bound constraint $0\leq x_{ij}\leq u_{ij}$, and the mass balance constraints 
\begin{equation}
\sum\limits_{(j, i) \in A^-(i)}x_{ji}-\sum_{(i, j) \in A^+(i)}x_{ij} = 0, \forall i\in N - \{s, t\}.
\end{equation} 

The value $|x|$ of a flow $x$ is the net flow into the sink;  that is, 
$|x| = \sum_{(i,t)\in \bar A(t)}x_{it}.$
The \textit{maximum flow problem} is to determine a feasible flow with maximum flow value.

Given a flow $x$, the \textit{residual capacity} $r_{ij}(x)$ of an arc $(i, j)\in A$ is 
$r_{ij}(x)=u_{ij}+x_{ji}-x_{ij}$.  It can be interpreted as
the maximum additional flow that can be sent from node $i$ to node $j$ using the arcs $(i, j)$ and $(j, i)$.  Under the assumption that at most one of $x_{ij}$ and $x_{ji}$ is positive,   $r_{ij}(x)+ r_{ji}(x)  = u_{ij}+u_{ji}$. Often, we will express the residual capacity of $(i, j)$ more briefly as $r_{ij}$. 

We refer to the network  consisting of the arcs with positive residual capacities as the \textit{residual network}, which we denote as $G(x)$.

%
%
%

Our algorithm uses preflows, which were introduced by Karzanov \cite{K74}.  In a preflow $x$, for each node $i \in N\setminus \{s,t\}$, the {\em excess} of node $i$ is defined to be 


$$e_x(i)= \sum_{(j,i)\in A^-(i)}x_{ji}-\sum_{(i, j)\in A^+(i)}x_{ij}.$$  

%
In a preflow $x$, the conservation of flow constraint is replaced by the following constraint on excesses.

\begin{equation}
e_x(i) \ge 0, \forall i\in N - \{s\}.
\end{equation}

Often, we will write the excess of node $i$ more briefly as $e(i)$.

For a preflow $x$, the value $|x|$ is the net flow into the sink node $t$.   
A maximum $s$-$t$ preflow is a preflow $x$ that maximizes the amount of flow arriving at node $t$.  Any preflow can be transformed into a flow by utilizing flow decomposition to return the flow from nodes with excess to the source node $s$.   (See, for example, \cite{AMO93}.)  The transformation of a preflow into a flow does not alter the amount of flow arriving at the sink.   This establishes the following well known lemma. 

\begin{lem}
\label{lem:maxpreflow}
The maximum value of an $s$-$t$ preflow is equal to the maximum value of an $s$-$t$ flow.
\end{lem}

We now make an additional assumption.  We assume that for any node $j \neq s$ or $t$, there are arcs $(j, s)$ and $(t, j)$ in $A$ with $u_{js} = u_{tj} = U$.   Adding these arcs is without loss of generality.  If $x^*$ is an optimal flow, then  $x^*$ can be efficiently transformed  into another optimum flow $x'$ such that for all $j \in N$, 
$x'_{js} = x'_{tj} =0$.  To transform $x^*$ into  $x'$, first express $x^*$ using flow decomposition (See \cite{AMO93}.)  Then eliminate all flows around cycles. 

 The advantage of introducing arcs directed into node $s$ is that any node $j$ with $d(j) = n+1$ can send all of its excess to node $s$ via the admissible arc $(j, s)$.   Accordingly,   $d(j) \le n+1$ throughout the algorithm. The advantage of introducing arcs directed is that these arcs needed in the initialization of the Enhanced LMES Algorithm, which is described in Subsection \ref{subsec:initialization}.  



Our algorithm is based on the stack-scaling algorithm of \cite{AOT89}, which in turn is based in the push/relabel algorithm of Goldberg and Tarjan \cite{GT88}.  We review the push/relabel algorithm next.

\section{Review of the generic Push-relabel algorithm}

In this section, we review the generic push-relabel algorithm of Goldberg and Tarjan  \cite{GT88}. 
Push-relabel algorithms begin with an initial preflow $x$ in which each arc out of node $s$ is saturated.  That is,
$$
x_{si} = \begin{cases}
u_{si}  & \text{ for each arc } (s,i) \in A^+(s)\\
0 & \text{ otherwise}.
\end{cases}$$

For each node $i \in N$, the algorithm maintains a \textit{distance label} $d(i)$, where $d(i)$ is a non-negative integer.
The distance labels are said to be \textit{valid} if they satisfy the following invariant.

\begin{inv}
(Validity) $d(t) = 0$; for all arcs $ (i, j) \in A$, if $r_{ij} > 0$, then $d(i)\leq d(j) + 1$.
\label{inv:validity}
\end{inv}

\bigskip

The generic push-relabel algorithm begins with the following distance labels: $d(s) = n$, and $d(i) = 0$ for $i \neq s$.  Throughout the rest of the generic algorithm, $d(t) = 0$, and $d(s) = n$.

The {\it length of a directed path} in $G(x)$ is the number of arcs on the path.  The distance labels $d(\,)$ provide lower bounds on the minimum length of a path to node $t$, as stated in the following lemma proved in \cite{GT88}. 

\begin{lem}
\label{lem:distance}
Suppose that the distance labels are valid with respect to a preflow $x$. Then for each $j \in N$, $d(j)$ is a lower bound on the length of the shortest path in $G(x)$ from $j$ to $t$.
\end{lem}
An arc $(i,j)$ is \textit{admissible} if $d(i)=d(j)+1$, and \textit{inadmissible} otherwise. The generic push-relabel algorithm maintains a preflow, and all pushes are on admissible arcs.    

We now describe the push-relabel algorithm in more detail.\\

If $e(i) > 0$, we refer to node $i$ as \textit{active}. The presence of active nodes in a push-relabel algorithm indicates that the solution is not yet a feasible flow.   The algorithm moves excess flow from active nodes through admissible arcs towards the sink. Excess flow at node $j$ that cannot be moved to the sink is returned to the source $s$ when  $d(j)= n+1$. The algorithm terminates with a maximum flow when the network contains no active node, i.e., when the preflow is also a flow. 
The generic push-relabel algorithm consists of repeating the following two steps until no node is active.\\

\noindent $Push(i, j)$\\ 
\indent   Applicability:  node $i$ is active and arc $(i, j)$ is admissible.\\
\indent  Action:  send min$\{e(i), r_{ij} \}$ units of flow in arc $(i,j)$.\\

\noindent $Relabel(i)$ \\
\indent   Applicability:  node $i$ is active and no arc of $A^+(i)$ is admissible.\\
\indent  Action: replace $d(i)$ by $d(i) + 1$. \\
\indent\indent (An alternative is to replace $d(i)$ by $1 + \min \{ d(j): \text{there is an admissible arc } (i, j) \}$.) \\

In the push operation, a push of $r_{ij}$ units of flow is called {\it saturating}. Otherwise, it is called {\it non-saturating}.  A push of $e(i)$ units of flow is called {\it emptying}.

The following results are either proved in \cite{GT88} or are implicit in that paper.

\begin{lem}
\label{lem:useful}
The generic push-relabel algorithm has the following properties:
\begin{enumerate}
\item At each iteration, $0 \le d(i) \leq n+1$ for each  node $i \in N$.
\item A distance label never decreases, and it increases at most $n+1$ times. The total number of relabel operations is less than $n^2$.
\item The total number of saturating pushes is less than $nm$.
\item The total number of pushes is $O(n^2m)$.
\item  The algorithm terminates with a maximum flow.
\item One can implement the push-relabel algorithm to find a maximum flow in $O(n^2m)$ time. 
\end{enumerate}
\end{lem}

The running time of many implementations of the push-relabel algorithm are the same as the total number of pushes.   To achieve these running times, one needs an efficient method of selecting admissible arcs.   Goldberg and Tarjan \cite{GT88} provided an efficient method that is sometimes referred to as the ``Current Arc'' method.   In seeking an admissible arc directed out of node $i$, the algorithm scans the arc list $A^+(i)$ in sequential order starting with its first arc.  
The scan terminates when an admissible arc $(i, j)$ is found.  
Then the algorithm sets CurrentArc$(i)$ to $(i, j)$, and $(i, j)$ becomes the starting point for the next scan of $A^+(i)$.  
If the end of $A^+(i)$ is reached without identifying an admissible arc, then there is no admissible arc in $A^+(i)$.  The push-relabel algorithm then increases $d(i)$ by 1 (a relabel step), and sets CurrentArc$(i)$ to be the first arc of $A^+(i)$. 

There are a variety of efficient implementations of the generic push-relabel algorithm.  Efficiencies can be obtained by use of Sleator and Tarjan's \cite{ST83} dynamic tree data structure and/or by a modification of the rule for selecting a node for a push/relabel. The $O(nm \log_k n + kn^2)$ max flow algorithm developed in this paper is based on the stack-scaling algorithm of Ahuja et al.~\cite{AOT89}, which is described in the next section. The stack-scaling algorithm is based on  Ahuja and Orlin's excess-scaling algorithm (\cite{AO89} and \cite{AMO93}).

\section{Large-Medium Excess Scaling Algorithm (LMES)}

In this section, we provide a modified version of the stack-scaling algorithm of \cite{AOT89}. We refer to this modified algorithm, which no longer relies on stacks, as the Large-Medium Excess Scaling (LMES) Algorithm.


In general, a scaling algorithm for the max flow problem solves the max flow problem as a sequence of scaling phases.   Associated with each phase is a scaling parameter $\Delta$, which remains constant during each phase and is decreased following the phase.

For the LMES Algorithm, the initial value of $\Delta$ is the least power of 2 that exceeds $U$.  
At each iteration of the $\Delta$-scaling phase for the LMES Algorithm, the algorithm maintains a preflow in which the excesses satisfy the following property.
 
\begin{equation}
\label{eq:parameter}
\textit{For each } i \in N\backslash \{s,t \}, \, 0 \le e(i) \le \Delta. 
\end{equation}  

The $\Delta$-scaling phase ends with the following property: for each $i \in N\backslash \{ s,t \}$,  
$0 \le e(i) < \frac{\Delta}{k}$ .   \\

After the scaling phase terminates, the scaling parameter is divided by the {\it scaling factor} $k$, and the next scaling phase begins.  The parameter $k$ is a power of 2 that is selected by the user.   The  algorithm continues until $\Delta < 1$, at which point the algorithm terminates with a maximum flow.  The total number of scaling phases is $O(\log_k U)$. 

We express most of our algorithmic running times and bounds using the parameter $k$.  In  theorems that give the running time for the LMES and the Enhanced LMES, we will specify a value of  $k$ that optimizes the worst case asymptotic running time.   

In the $\Delta$-scaling phase, Equation \ref{eq:parameter} requires that $e(i) \le \Delta$ for each node $i$.   To maintain this inequality, we modify the push operation.\\

 \noindent \textit {LMES-Push}$(i, j)$ \\
\indent   Applicability:  Node $i$ is active and arc $(i, j)$ is admissible.\\
\indent  Action:  send $\delta = \min\{ e(i), r_{ij}, \Delta - e(j) \}$ units of flow in arc $(i,j)$.\\

 In the $\Delta$-scaling phase, a node $i \in N \backslash \{s,t\}$ is called a \textit{large-excess node} if $e(i)\geq\frac{\Delta}{2}$.  We let {\it LargeSet} denote the subset of nodes in $N\backslash \{s, t\}$ with large excess.  Node $i$   is called a \textit{medium excess node} if 
$\frac{\Delta}{k}\leq e(i) <\frac{\Delta}{2}$.  We let {\it MediumSet} denote the subset of nodes in $N\backslash \{s, t\}$ with medium excess.

The procedure  {\it LMES-Select-Node}, described below, selects  a large or medium excess  node $i$.  \\

\noindent $i :=$ {\it LMES-Select-Node} \\ 
\indent  Applicability:  there is a node with medium or large excess.\\
\indent  Action: if LargeSet $\neq \emptyset$, then
$i := \text{argmin} \{d(j) : j \in \text{LargeSet} \}$; \\
\indent\indent   else,  
$i := \text{argmax} \{ d(j) : j \in \text{MediumSet}\} $.\\

The phase ends when there are no large or medium excess nodes.  Otherwise, the algorithm selects a medium or large excess node $i$ using the above procedure.    If there is an admissible arc $(i, j)$, then the algorithm next runs procedure \textit{LMES-Push}$(i, j)$.   If there is no admissible arc, the algorithm runs \textit{Relabel}$(i)$.    Subsequently, the algorithm returns to selecting a large or medium excess node.  \\

%
%

The following theorem is stated and proved in \cite{AOT89}.  

\begin{thm}
\label{th:LMEStime}
The LMES Algorithm terminates after $O(\log_k U)$ scaling phases with a maximum flow from source to sink.  Its running time is $O(nm + kn^2 + n^2 \log_k U)$.
\end{thm}


In order to balance the terms in the running time and optimize the overall running time, Ahuja et al. \cite{AOT89} chose $k$ to be the least power of 2 that exceeds  $2 + \frac{\log U}{\log\log U}$.  For this value of $k$, the running time of the LMES Algorithm is $O(nm + n^2 \frac{\log U}{\log\log U})$. 

\bigskip

\section{The running time analysis of the LMES Algorithm}
\label{sec:LMES}

In this section, we will  analyze the running time of the LMES Algorithm and prove Theorem \ref{th:LMEStime}.   
The proof here is based on the proof in \cite{AOT89}.   We include the proof here for two reasons.   First, the LMES Algorithm is different from the stack scaling algorithm, most notably because the LMES Algorithm does not use stacks.  Second, the potential function arguments here are the starting point for the analysis of the Enhanced LMES Algorithm, whose pseudo-code is first presented in Section \ref{sec:annotated}.  

\subsection{The number of large pushes}   The proof of  Theorem \ref{th:LMEStime} relies on two potential function arguments.   The first potential function is $\Phi_1$.

  $$\Phi_1:=\sum_{i\in N\backslash \{s,t\}}\frac{e(i)d(i)}{\Delta}.$$
  
  \smallskip

Ahuja and Orlin \cite{AO89} relied on $\Phi_1$ in establishing the running time of the excess scaling algorithm.   The excess-scaling algorithm is the special case of the LMES Algorithm in which there are no medium pushes, and $k = 2$.  The running time of the excess-scaling algorithm is $O(nm + n^2 \log U)$.    

Ahuja et al. \kern -3pt \cite{AOT89} repurposed the potential function as part of the proof of the running time of the stack scaling algorithm.

\begin{lem}
During the $\Delta$-scaling phase of the LMES Algorithm, the total amount of flow sent in pushes  is less than $2n^2\Delta$.
\label{lem:mostpush}
\end{lem}
\begin{proof}

Each push of $\delta$ units of flow during the $\Delta$-scaling phase decreases $\Phi_1$ by $\frac{\delta}{\Delta}$.  To prove the lemma, we will prove that the total decrease in $\Phi_1$ (that is, the sum of all of the decreases in  $\Phi_1$) during the $\Delta$-scaling phase is at most  $2n^2$.   

The total decrease in  $\Phi_1$ during the $\Delta$-scaling phase is bounded above by the  value of $\Phi_1$ at the beginning of the $\Delta$-scaling phase plus the  total increase in  $\Phi_1$ during the $\Delta$-scaling phase.

The initial value of  $\Phi_1$ is less than $(n-2)(n+1)$ because the sum is over at most $n-2$ nodes, and for each node $i$, $e(i) < \Delta$ and  $d(i) \le n+1$.   To bound the total increase in $\Phi_1$ during the scaling phase, we observe  that the only way for $\Phi_1$ to increase during the phase is when there is a distance relabel. For each node $ i \in N$, an increase of $d(i)$ by 1 leads to an increase in $\Phi_1$ by at most 1.  Since $d(i) \le n+1$,  the total increase in $\Phi_1$ caused by increases in distance labels is less than $(n-2)(n+1)$ over all phases.

Thus, the total amount pushed during the $\Delta$-scaling phase is at most $\Delta$ times the initial value of $\Phi_1$ plus $\Delta$ times the total increase of $\Phi_1$, which is less than $2n^2\Delta$.		
\end{proof}	
	
By Lemma \ref{lem:useful}, the number of saturating pushes is $O(nm).$   We are interested in bounding the number of non-saturating pushes.   

We say that a push of $\delta$ units of flow in $(i, j)$ is a \textit{large push} during the  $\Delta$-scaling phase if it is non-saturating and if  $ \delta \ge \frac{\Delta}{2}$.  
We say that the push in $(i, j)$ is a \textit{medium push} if it is non-saturating and if $\frac{\Delta}{k} \le \delta < \frac{\Delta}{2}$.  

The following is an immediate consequence of Lemma \ref{lem:mostpush}.   

\begin{cor}
The total number of large pushes during the $\Delta$-scaling phase is at most $4n^2$.  The total number of large pushes over all scaling phases of the LMES Algorithm is at most $4n^2 \log_k U$.
\label{cor:largepushes}
\end{cor}


In the next subsection, we will bound the number of medium pushes.   But before doing so, we provide some insight as to why the use of medium pushes leads to a speed-up of the excess-scaling algorithm.  

If there were only large pushes, the algorithm would be equivalent to the Ahuja-Orlin 
excess-scaling algorithm \cite{AO89}, in which $k = 2$.   In that case, the $\Delta$-
scaling phase would terminate with excesses bounded above by $\frac{\Delta}{2}$.  
The increase in $\Phi_1$ because of the beginning of each scaling phase is 
$O(n^2)$ per scaling phase and $O(n^2 \log U)$ over all scaling phases.  The 
increase in $\Phi_1$ due to all relabels is $O(n^2)$.   The running time of the excess 
scaling algorithm  $O(nm + n^2 + n^2\log U)$.   One question that this analysis 
raised is the following:  is there a way to modify the excess scaling algorithm so that 
the $ n^2\log U$ term in the running time is decreased.   It would be OK if the $ n^2$ 
term in the running time increased provided that it didn't increase much.  The LMES Algorithm answered this question in the positive.

In the LMES algorithm, the medium pushes are of size proportional to $\frac{\Delta}{k}$.  At the end of the $\Delta$-scaling phase, the excesses are all less than $\frac{\Delta}{k}$, which permits an increase in the scaling factor from 2 to $k$.   Increasing the scale factor improves the running time by a factor of $O(\log k)$ provided that there are not too many medium pushes in total.   Selecting medium pushes from nodes with maximum distance labels is sufficient for establishing the bound on medium pushes proved in the next subsection.

\subsection{The number of medium pushes}  We  now use a different potential function argument to bound the number of medium pushes.  We state the bounds as a lemma.

\begin{lem}
The total number of medium pushes in the LMES Algorithm is $O(kn^2 +n^2 \log_k U)$.
\label{lem:mediumpushes}
\end{lem}
\begin{proof}
We use  a potential function $\Phi_2$ of \cite{AOT89}, which is based on parameters  $\ell$ and $P$.  
At the beginning of a scaling phase, $\ell=n+1$. Subsequently, 
$\ell$ is the minimum distance label from which there has been a medium push during the phase.    At the beginning of the scaling phase, $P = \emptyset$.  Subsequently,

\indent $P:=\{i\in N: d(i)>\ell\} \, \cup 
 \{i \in N:d(i)=\ell$, and there was a medium push from $i\}$. \\
 
 We observe that whenever there is a medium push from any node $v$, then $v \in P$ except possibly for the first medium push from  node $v$.  For any scaling phase, the  number of first medium pushes from nodes of $N$ is at most $n-2$.  To bound the remaining number of medium pushes, we rely on the following potential function  developed in  \cite{AOT89}.

The potential function $\Phi_2$ is defined as follows.

 $$\Phi_2:=\sum_{j\in P}  e(j)\cdot\frac{d(j)-\ell+1}{\Delta}.$$ 

Note that $\Phi_2 = 0$ at the beginning of a scaling phase.
 
Every medium push from a node in $P$ leads to a decrease in $\Phi_2$  by at least $\frac{1}{k}$.  
Thus, the total number of medium (non-first) pushes in the $\Delta$-scaling phase is 
at most $k$ times the total decrease in
$\Phi_2$ during the phase.  Because $\Phi_2 = 0$ at the start of the scaling phase, the total  decrease in
$\Phi_2$ during the phase is bounded by the total increase in
$\Phi_2$ during the phase. 

There are three ways in which $\Phi_2 $ can increase during a scaling phase:
\begin{enumerate}
\item A node's distance label increases.
\item $\ell$ decreases.
\item A node enters $P$.  
\end{enumerate}

We first consider the impact  of relabels on $\Phi_2$.
At an iteration at which a node $i$ increases its distance label by 1,  $\Phi_2 $ increases by 
$\frac{e(i)}{\Delta}$, which is at most 1, except possibly in the case in which $d(i) = \ell$ and $i\notin P$.  In this latter case, $\Phi_2 $ increases by  $\frac{2e(i)}{\Delta}$, which is at most 2.  Thus the increase in $\Phi_2 $  over all phases due to distance increases is at most 
$2(n-2)(n+1)$.

We next consider the impacts from decreases in the parameter $\ell$.  Suppose that $\ell$ is decreased by 1.  Let $P$ be defined immediately prior to the decrease in $\ell$.   
For each node $i \in P$, $d(i) - \ell$ increases by 1, and  
 $\Phi_2$ increases by at most 1.  The total increase in $\Phi_2$ per unit decrease in  $\ell$ is at most $|P|$, which is at most $n-3$.
The increase in  $\Phi_2$ caused by all of the decreases in $\ell$ is less than $ n^2$.



Finally, we consider the increases in $\Phi_2$ that occur when nodes enter $P$.  If a node $i$ with $d(i) = \ell$ enters $P$, then $\Phi_2$ increases by at most 1.   If a node $i$ enters $P$ because of a decrease of $\ell$ by 1, then  $\Phi_2$ increases by at most 2.  (We consider $\ell$ as decreasing by 1 unit at a time.)  Thus, the increase in $\Phi_2$ caused by nodes entering $P$ is less than $2n$.   

By the above, the total increase in $\Phi_2 $ over all scaling phases is $O(n^2 + \frac{n^2 \log_k U}{k} )$, and thus the total number of medium pushes over all scaling phases is  $O(kn^2 + n^2 \log_k U )$.
\end{proof}

To complete the proof of Theorem \ref{th:LMEStime}, we need to bound the time for the selection of medium and large excess nodes.   We first state it as a lemma.

\begin{lem}
The total  time for selecting nodes is 
$O( \# \text{of pushes} +  \# \text{of relabels} + n \log_k U)$.
\label{lem:selection}
\end{lem}

\begin{proof}   In order to carry out \textit{LMES-Select-Node} efficiently, the algorithm maintains the following data structures.  These data structures differ from those in \cite{AOT89} because the selection rule in LMES is different from the rule in the stack-scaling algorithm. 

\begin{itemize}
\item  MSet$(i) =$ linked list of medium excess nodes with distance label $i$ for $0 \le i \le n+1$.
\item  LSet(i) = linked list of large excess nodes with distance label $i$ for $0 \le i \le n+1$.
\item MaxML = maximum distance label of a node that has medium or large excess.  If there is no medium or large excess node, then the scaling phase terminates.
\item  MinL = minimum distance label of a node with large excess. If there is no large excess node, then MinL = -1.
\item  NextL$(i)$ for $0 \le i \le n-1$.    If LSet$(i) \neq \emptyset$,   then NextL$(i) = \min \{j : j > i$ and LSet$(j) \neq  \emptyset \}$.  If there is no such $j$ or if LSet$(i) = \emptyset$, then NextL$(i) = 0$.
\end{itemize}

We select a node in $O(1)$ time as follows:  if MinL $\neq 0$, we select the first node in LSet(MinL); else, if MaxML $\neq 0$, we select the first node in MSet(MaxML); if MaxML = 0, the phase ends.


The total time to initialize the data structures is $O(n)$ per scaling phase.  We note that MSet$(\,)$ and LSet$(\,)$ can each be updated in $O(1)$ time following a push/relabel.  To establish the running time of Lemma \ref{lem:selection}, we will show how to carry out all of the updates of  MaxML, MinL, and NextL$(\,)$ in time proportional to the number of pushes plus the number of relabels. 

We first consider increases in MaxML.   MaxML can increase following a relabel of a medium or large excess node.   This update takes $O(1)$ time per relabel and $O(n^2)$ over all scaling phases.    

We now consider decreases in MaxML.   MaxML can decrease by at most 1 following a medium or large push.   However, following a saturating push from node $v$, MaxML can decrease by much more than 1.  If MaxML decreases from a value $i$, the algorithm updates MaxML by iteratively checking MSet$(j)$ and LSet$(j)$ for $j = i-1, i-2, \dots$, until it finds a non-empty set.  (If all sets are empty, the phase terminates.) The time to update MaxML when MaxML is decreasing is proportional to the amount of decrease.  The total decrease in MaxML during a scaling phase is bounded from above by the total increase in MaxML plus $n$. Thus, the total decrease in MaxML over all scaling phases is bounded by $n^2 + n \log_k U$. The time to carry out all updates MaxML over all scaling phases is $O(n^2 + n \log_k U)$. 

We next consider updates to MinL and NextL$(\,)$.  
There are three operations (and four cases) in which NextL$(\,)$ or MinL needs to be updated.

\begin{itemize}
\item The beginning of a new scaling phase.
\item A medium or large push in $(v, w)$ causes node $w$ to have large excess. 
\item  A large push in $(v, w)$ eliminates the large excess from node $v$.
\item There is a relabel of a large excess node $v$.
\end{itemize}

At the beginning of a scaling phase, the time to initialize NextL$(\,)$ and MinL is $O(n)$. 
For each of the three remaining cases, NextL( )  and MinL can be updated in $O(1)$ steps.  Thus, the total time to update NextL( )  and MinL is 
$O(\# \text{ of pushes} + n \log_k U)$.  This completes the proof of Lemma \ref{lem:selection} as well as the proof of Theorem \ref{th:LMEStime}.
\end{proof}

We observe that we cannot use an amortized analysis similar to the analysis of MaxML when analyzing MinL because one cannot provide simple bounds on the increase or decrease in MinL over a phase.   If a medium push from level $j$ creates a large excess node, then MinL changes from $-1$ to $j - 1$.   If a saturating push from a node at level $j$ results in LargeSet becoming empty, then MinL changes from $j$ to $-1$.   For this reason, we introduced the array NextL$(\,)$.

\section{The parameters $Q$, $\epsilon$, and $M$}

We will employ a number of parameters in the remainder of this paper.  We first define parameters $Q$, $\epsilon$, and $M$.   

We let $Q =   \lceil \log_k 4n \rceil$.   We let $\epsilon = k^{-Q}$.  We let $M = k^{2Q}$. 

The following are true: 
\begin{enumerate}  
\item	$\frac{1}{4nk} < \epsilon \le \frac{1}{4n}$.
\item	$Q$ is the least integer for which $k^{-Q} \le \frac{1}{4n}$.
\item   $M = \epsilon^{-2}$.
\item $16n^2 \le M < 16k^2M^2$.  
\item	If $\Delta$ is the scaling parameter for a scaling phase, then $Q$ scaling phases later, the scaling parameter is $ \epsilon \Delta$.
\end{enumerate}  

 In subsequent sections, many of the parameters will be expressed as powers of $\epsilon$ or expressed in terms of multiples of $M$ or  $Q$.

\section{An overview of the Enhanced LMES Algorithm}

In this section, we present concepts and procedures that improve the running time from 
$O(nm + kn^2 + n^2 \log_k U)$ to  $O(kn^2 + nm \log_k n)$.  The $O(kn^2)$ term is due to the increases in $\Phi_2$ from distance relabels.   That term also occurs in the Enhanced LMES Algorithm.   We do not consider that term further in this section.  
Instead, we focus on the $O(n^2 \log_k U)$ term, which arises in the analysis of both potential functions. 

\begin{enumerate}
\item \textbf{Increases in $\Phi_1$ at the beginning of phases.}   Consider each node $i$.    When the scaling parameter is decreased by 
a factor $k$ at the beginning of a scaling phase, the contribution $\frac{e(i) d(i)}{\Delta}$ increases by a factor 
of $k$.   This leads to an increase in $\Phi_1$  by $O(n)$.   Over $log_k U$ scaling phases, the contribution 
to the increase in $\Phi_1$   is $O(n \log_k U)$ for each node $i \in N$, which accounts for $O(n \log_k U)$ 
large pushes.  Summing over all nodes, this accounts for $O(n^2 \log_k U)$  large pushes.   
\item  \textbf{Increases in $\Phi_2$ due to decreases in $\ell$.}   Consider each node $i$.   For each unit decrease in the parameter $
\ell$, each node in $P$ leads to an increase in $\Phi_2$  by $O(\frac{1}{k})$.  Thus, for each node $i$, 
reductions in $\ell$ lead to an increase in $\Phi_2$  of $O( \frac{n}{k} )$ during the scaling phase, accounting for 
$O(n)$ medium pushes. Node $i$ accounts for $O(n \log_k U)$ medium pushes over all scaling phases.   
Summing over all nodes, this accounts for $O(n^2 \log_k U)$  medium pushes.
\end{enumerate}

Any further discussion of increases in $\Phi_1$ in this section refers to increases due to the beginning of phases.  Any further discussion of increases in  $\Phi_2$  in this section, refers to increases caused by decreases in $\ell$.  

In the Enhanced LMES Algorithm, we reduce this increase of $\Phi_1$  to $O(nm \log_k n)$.   We reduce this increase in $\Phi_2$  to $O(\frac{nm \log_k n}{k})$.  These improvements reduce the number of large and medium pushes to  
$O(kn^2 + nm \log_k n)$  in total.  

We will rely on various concepts, procedures, and data structures to achieve the improved running time.   They are:
\begin{enumerate}
\item   the partition of arcs into small, medium, and large arcs;
\item   abundant arcs and contractions;
\item   special nodes; and
\item   the flow-return forest.
\end{enumerate}
 
In the next four subsections, we provide an overview of how these four ideas lead to an improved strongly polynomial running time.

\subsection{Small, medium, and large arcs}
 
We say that arc $(i, j)$ is \textit{small} at the $\Delta$-scaling phase 
if $u_{ij} + u_{ji} <  \epsilon^5 \Delta$. We say that the arc $(i, j)$ is \textit{medium} 
at the $\Delta$-scaling phase 
if $\epsilon^5 \Delta \le u_{ij} + u_{ji} < 2M \Delta$.    
We say that the arc is arc $(i, j)$ is  \textit{large} 
at the $\Delta$-scaling phase 
if $ u_{ij} + u_{ji} \ge 2M \Delta$.  
We sometimes refer to the arc as small, medium, or large without reference to $\Delta$ 
if the scaling parameter is clear from context.

The small arcs do not contribute much to the excess of a node, and accordingly do not contribute much to the increases in $\Phi_1$  or $\Phi_2$  described above.   With respect to the analysis of the potential functions, we can (essentially) ignore flows in small arcs.

We now consider medium arcs.  Each arc $(i, j)$ is a medium arc for $O(\log_k n)$ scaling phases.  
Let $\Psi = \{\Delta: \Delta \text{ is a parameter for a scaling phase} \}$.   For a given scaling parameter $\Delta \in  \Psi$, let $N_{med} (\Delta)$ denote the set of nodes that are incident to at least one medium arc during the $\Delta$-scaling phase. Accordingly,

$$\sum_{\Delta \in \Psi} |N_{med}(\Delta)| = O(m \log_k n).$$

Moreover, the contribution to the increases in $\Phi_1$  and $\Phi_2$  of nodes that are incident to a medium arc are 
$O(nm \log_k n)$ and $O(\frac{nm \log_k n}{k})$ respectively.   Thus, over all scaling phases, nodes that are incident to medium arcs account for $O(nm \log_k n)$ large and medium pushes.

What remains is to analyze the contribution to increases in $\Phi_1$  and $\Phi_2$  from nodes that are not incident to any medium arc.  But prior to focusing on nodes that are not incident to medium arcs, we discuss abundant arcs and contractions.

\subsection{Abundant arcs and contraction}

We say that an arc $(i, j)$ is abundant at the $\Delta$-scaling phase if $(i, j)$ is large and if 
$r_{ij} \ge M \Delta$, where $r_{ij}$ is the residual capacity at the beginning of the scaling phase.   Once an arc becomes abundant, it remains abundant.   We state this result as the following corollary to Lemma \ref{lem:totalflow}, which is the counterpart of Lemma \ref{lem:mostpush} for the Enhanced LMES. \\

\noindent \textit{\textbf{Corollary \ref{cor:abundance}}.
If $(i, j)$ is abundant at the beginning of the $\Delta$-scaling phase of the Enhanced LMES Algorithm, then for every scaling parameter $\Delta' \le \Delta $ the following are true. 
\begin{enumerate}
\item   $(i, j)$ is abundant in the $\Delta'$-scaling phase and
\item   $r_{ij} > 0$ throughout the $\Delta'$-scaling phase.
\end{enumerate}}


The importance of abundant arcs is that they can be used to transform the max flow problem into an equivalent problem of smaller size.  In particular, if there is a directed cycle $C$ of abundant arcs, then one can contract the cycle $C$ into a single node called a \textit{merged node} and run the Enhanced LMES algorithm on the contracted graph.   
As each abundant cycle is discovered, it is contracted into a merged node.   The Enhanced LMES Algorithm  continues from where it left off  on the contracted graph.   Eventually, the algorithm determines a maximum flow in some contracted version of the original network.   Then the maximum flow in this contracted graph is converted into a maximum flow in the original graph by expanding the merged nodes in the reverse order in which they were contracted.   We describe the contraction  in Section \ref{sec:contraction}.  We describe the expansion in the appendix in Section \ref{sec:expansion}.    

Contraction, in and of itself, does not even lead to a strongly polynomial time algorithm.  We illustrate this fact by the following example.  Suppose that the node set 
$N = \{s, 1, 2, t\}$.   Suppose that there are arcs $(s, 1)$ and $(1, t)$ with a capacity of $U = k^\alpha $ for some very large integer $\alpha$.  There are also arcs $(s, 2)$ and $(2, t)$ with a capacity of $1$.   After initialization, $e(1) = k^\alpha $ and $e(2) = 1$.   The initial scaling parameter is $\Delta = U$.   However, there is no push from node 2 until $\Delta = 1$. If each scaling parameter is obtained by dividing by $k$, then the number of scaling phases would be $\alpha$, which is not strongly polynomial. 

The algorithm fails to be strongly polynomial in this example because of the scaling phases in which there was no push.  We refer to a scaling phase as \textit{useful} if there is at least one push during the scaling phase or if the flow-return forest (discussed in Subsection \ref{subsec:FRF1}) is non-empty.  Otherwise, it is \textit{useless}.   Our algorithm avoids useless scaling phases by modifying the rule for choosing the next scaling parameter.  We will establish the following result.\\

%

\noindent \textit{\textbf{Corollary \ref{cor:phases}.}   The number of scaling phases in the enhanced LMES is 
$O(m \log_k n)$.}   \\

\subsection{Special nodes}

The enhanced LMES will partition the nodes in $N\setminus \{s,t\}$ into nodes that are ``special'' and nodes that are ``non-special''.   (Actually, we will further partition the non-special nodes into two parts, but for now it suffices to consider special and non-special nodes.) We will define the term ``special'' later in this subsection. For now, we provide two properties of special nodes and one property of non-special nodes. 

\begin{enumerate}
\item   If node $i$ is special, then $e(i) < 1.5 \epsilon ^4 \Delta$ at the beginning of the scaling phase.
\item   If node $i$ is special, then $e(i) < 1.5 \epsilon ^4 \Delta $ whenever $\ell$ is reduced.
\item   If node $i$ is not special, then within $O(\log_k n)$ scaling phases, $i$ will be incident to a medium arc, or $i$ will be a node of an abundant cycle that is contracted.
\end{enumerate} 

The first two properties imply that each special node contributes much less than $\frac{1}{n}$ per scaling phase to the increase of $\Phi_1$ and $\Phi_2$ .    Since there are $O(m \log_k n)$ scaling phases in total, special nodes account for $O(nm \log_k n)$  large and medium pushes.

The third property implies that the total number of nodes that are not special is 
$O(m \log_k n)$ over all scaling phases.   Thus, nodes that are not special will account for $o(nm \log_k n)$  large and medium pushes.  \\

To summarize, non-special nodes account for $O(kn^2 + nm \log_k n)$ medium and large pushes.   And special nodes also account for $O(kn^2 + nm \log_k n)$ medium and large pushes.   Thus, the number of medium and large pushes is $O(kn^2 + nm \log_k n)$.  This (apparently) leads to the running time claimed in our paper.   However, we have not yet mentioned an important issue that arises because of special nodes.    In order to ensure that Properties 1 and 2 are satisfied by each special node, we permit the excess for a special node $i$ to be negative, provided that 
$e(i) \ge -1.5 \epsilon ^4 \Delta$.    This small difference in what is permitted makes a huge difference in the analysis.    The generic push/relabel algorithm terminates with a maximum flow because every excess is non-negative at termination.   We needed to adapt our algorithm so that every excess is 0 at termination.  This required new methods for sending flow to nodes with negative excess. 

In the next subsection, we will describe a new data structure that returns flow to nodes with negative excess.  Using this data structure, excesses are always at least $- \epsilon \Delta$, and all excesses are nonnegative at termination.   But first, we define what it means for a node to be special.  We also explain why it makes sense to permit negative excesses for special nodes.

Recall that $(i, j)$ is abundant if $(i, j)$ is large and $r_{ij} \ge M \Delta$.  If $(i, j)$ is abundant, and if $(j, i)$ is not abundant, then arc $(j, i)$ is called \textit{anti-abundant}.  If $(i, j)$ and $(j, i)$ are both abundant, then arc  $(i, j)$ is called \textit{bi-abundant}.   A bi-abundant arc is an abundant cycle with two arcs.  In general, the algorithm contracts any bi-abundant arc.   (There is an exception discussed in the section on the flow-return forest.)   

During the $\Delta$-scaling phase of the Enhanced LMES Algorithm, we  require that the residual capacity of each anti-abundant arc is a multiple of $\frac{\Delta}{k}$.   In Section \ref{subsec:newinvariant}, we explain how this invariant can be satisfied throughout the Enhanced LMES with only a constant increase in running time.  

For a given preflow $x$ and a given node $v$, we define the \textit{$\Delta$-imbalance} of $v$ with respect to $x$ to be:
$$ \text{IMB}(v, x, \Delta) = e_x(v) + \sum_{(j, v) \in Anti(\Delta)} r_{jv}(x) - 
	\sum_{(v, j) \in Anti(\Delta)} r_{vj}(x).$$ 

We define $\text{IMB}(v, \Delta)$ to be $\text{IMB}(v, x, \Delta)$, where $x$ is the initial preflow at the beginning of the $\Delta$-scaling phase. 

Let us for the time being ignore small capacity arcs.   (Assume that there are no small arcs incident to node $v$.)  Because residual capacities of anti-abundant arcs are multiples of $\frac{\Delta}{k}$, 

$$\text{IMB}(v, x, \Delta) \equiv e_x(v)  \left(\kern -10pt \mod \frac{\Delta}{k}\right ).$$

A node $v$ is called a \textit{special node} at the $\Delta$-scaling phase if $v$ is not incident to any medium arcs and if  $|\text{IMB}(v,  \Delta)| \le \epsilon^4 \Delta $.

If $v$ is special, then every medium or large push from $v$ will be a multiple of $\frac{\Delta}{k}$ units of flow.   Our algorithm analysis relies on special nodes having an excess close to 0 at the beginning of scaling phases and when $\ell$ is reduced.   This works well for a special node $v$ if  $0 \le \text{IMB}(v,  \Delta) \le \epsilon^4 \Delta $.   When $v$ is no longer medium or large, then 
$0 \le e(v) \le \epsilon^4 \Delta $.

Consider what happens if $ - \epsilon^4 \Delta < \text{IMB}(v,  \Delta) < 0 $.   Suppose that node $v$ is not permitted to have a negative excess. When $v$ is no longer medium or large, then 
$\frac{\Delta}{k} - \epsilon^4 \Delta < e(v) < \frac{\Delta}{k} $.  In this scenario, a special node would contribute significantly to the increase of the potential functions.

So, our Enhanced LMES Algorithm permits pushing flow from a special node $v$ when $e(v)$ is close to but slightly less than $\frac{\Delta}{k}$.  This can result in a very small negative excess. 
To ensure that flow is returned to $v$ prior to termination, the algorithm uses a data structure that we call ``the flow-return forest.''

\subsection{The flow-return network}
\label{subsec:FRF1}

Actually, a special node $v$ can have negative excess provided that $e(v) \ge -1.5 \epsilon ^4 \Delta$.   The procedures of the LMES Algorithm are not enough to guarantee that flow will be sent to node $v$.   Accordingly, we make two additional changes in the LMES Algorithm.  The first change is the requirement of a ``buffer'' excess at each non-special node.   In particular, if node $w$ is not special, then we ``require'' that $e(w) \ge \epsilon  \Delta$.   If $w$ is non-special and if $e(w) \ge \epsilon  \Delta$, we say that $w$ is \textit{normal}.   If $w$ is non-special and if $e(w) < \epsilon  \Delta $, we say that $w$ is \textit{violating}.   A node $w$ becomes violating at the beginning of the first scaling phase after $w$ stops being special. 

When a node $v$ becomes violating, it is added to a data structure that we call the ``flow-return forest'' (FRF), which we usually denoted as $F$.  The flow-return forest $F$ has the following properties.

\begin{enumerate}
\item   $F$ is a forest; that is, it has no cycles.
\item   Each component of $F$ has a root node.   For each $w  \in   F$, $Root(F, w)$ denotes the root node of $w$ in $F$.
\item   For each node $w  \in   F$, there is a directed path of abundant arcs in $F$ from $Root(F, w)$ to $w$.  This implies that  $F$ is a directed forest, with all arcs directed away from root nodes.
\item   An arc $(i, j)$ is called \textit{FRF-eligible} if $(i, j)$ is abundant and $d(j) \le d(i) + 1$.   Every arc of the 
flow-return forest is FRF-eligible.  
\item   If a violating node $v$ is added to $F$ in the $\Delta$-scaling phase, then Flow-Needed$(v)$ is set to 
$\frac{ \epsilon^2 \Delta}{k}$.   Suppose that
$\Delta' = \epsilon ^2\Delta$.  If $v$ is still violating at the $\Delta'$-scaling phase, then 
$\frac{ \epsilon^2 \Delta}{k}$ units 
of flow will be sent from $Root(F,v)$ to $v$ in $F$.  (Note that the reversal of the path $P$ from $Root(F,v)$ to 
$v$ will have residual capacity after flow is sent in $P$.   Condition (4) on distance labels is needed so that 
every arc of the reversal of $P$ is valid.) 
\item   Each root vertex $i  \in   F$ is a normal node.  A root vertex $i$ maintains additional excess called $Reserve(i)$, where
$$ Reserve(i) = \sum_{v \in Desc(i, F)} NeededFlow(v),$$
where $Desc(i, F)$ denotes the descendants of node $i$ in $F$.    The reserve requirements imply that $e(i) \ge \epsilon  \Delta + Reserve(i)$.

\item   Every leaf node of $F$ is violating.  (Non-violating leaf nodes are deleted from $F$).
\item   If $v  \in   F$ at the $\Delta$-scaling phase, then 
$e(v) \ge -\frac{\epsilon  \Delta }{2k}$.
\item   If $v  \in   F$ is violating at the $\Delta$-scaling phase and if Flow-Needed$(v) = \frac{\Delta}{k}$, then $\frac{\Delta}{k}$ units of flow are sent from $Root(F, v)$ to $v$ in $F$.  Subsequent to the flow being sent, $e(v) \ge \frac{\Delta}{2k}$.
\end{enumerate}

In ensuring that the above properties are all satisfied by the flow-return forest, we also need to address several complexities that arise.  We briefly describe them here.   

\begin{enumerate}
\item   The algorithm needs to be able to add any violating node to the flow-return forest while satisfying all of the above conditions.   The FRF-eligible Path Lemma of Section \ref{sec:FRF} states the following.  For any violating node 
$v$, there is a path of FRF-eligible arcs from a normal node to node $v$.   By relying on this lemma, we are always able to add violating nodes to the flow-return forest.  Using a breadth first search algorithm, the running time is $O(m)$ to add a violating node to $F$.
\item   The number of violating nodes added to the flow-return forest is $O(m)$.   In order to achieve a total running time of $O(kn^2 + m \log_k n)$, we cannot afford to use the $O(m)$ breath first search algorithm to add each violating node $v$.  Instead, we use a depth first search algorithm that runs in $O(n)$ time. 
\item  If node $j \in F$ and if $j$ is relabeled, then the FRF-eligible arc directed into node $j$ might stop being eligible.   We bypass this difficulty by ensuring that  no node of $F$ is relabeled.  
\item   The reserve of root nodes might exceed $\Delta$.    We bypass this difficulty by not including the reserve of root nodes when evaluating the potential functions.  
For each root node $v$, the term included in the evaluation of the potential functions is $e(v) - Reserve(v)$.   
\item   We need to ensure that $e(v) - Reserve(v) \le \Delta + k \epsilon  \Delta$.    To guarantee this property, we sometimes send additional flow from root nodes in the procedure \textit{FRF-Delete}.
\end{enumerate}

\bigskip

\section{Annotated Pseudo-code}
\label{sec:annotated}

Here we present a very high level pseudo-code for the Enhanced LMES Algorithm.  We also present  the pseudo-code of the scaling phase.  Following the pseudo-code is an annotation explaining of the steps and where the pseudo-code of  procedures are located.\\

\noindent \textbf{\textit{Enhanced-LMES Algorithm}}\\
01.   \textbf{begin}\\
02.   \hspace{.15 in}  \textit{Initialize-Enhanced-LMES}\\
03.   	\hspace{.15 in}  \textbf{while} $e(v) \neq 0$ for some $v \neq s$ or $t$,  \textbf{do} 
\textit{ScalingPhase}$(\Delta, F)$;\\
04.   	\hspace{.15 in}  expand the contracted cycles, obtaining a max flow;\\
05.   \textbf{end}\\

\noindent \textbf{Line 02.}   The initialization is very similar to the initialization of the LMES Algorithm except that each node other than $s$ and $t$ is guaranteed to have an excess of at least $\epsilon  \Delta$ after initialization.   Details are given in Subsection \ref{subsec:initialization}.\\

\noindent \textbf{Line 03.}    \textit{ScalingPhase}( ) is the main procedure.  A high level view is presented below.\\

\noindent\textbf{\textit{Procedure ScalingPhase$(\Delta, F)$}}\\
01.   \textbf{begin}\\
02.   	\hspace{.15 in}  \textbf{while} there is a newly violating node $v \notin F$, \textbf{do} 
\textit{FRF-Add}$(v, \Delta, F)$;\\
03.   \hspace{.15 in}  	\textbf{while} there is a violating node $v \in F$ 
with $FlowNeeded(v) = \frac{\Delta}{k}$,\\
--  \hspace{.50 in}		 \textbf{do} 	\textit{FRF-Pull}$(v, F, \Delta)$;\\
04.   	\hspace{.15 in}  \textit{FRF-Recursive-Delete-and-Merge}$(\Delta, F)$;\\
05.   	\hspace{.15 in}  \textbf{while} there is a medium or large excess node \textbf{do} \\
06.   	\hspace{.30 in}  	$v :=$ \textit{Enhanced-LMES-Select-Node};\\ 
07.   \hspace{.30 in}		\textbf{if } $v \in F$, \textbf{then do}
	\textit{FRF-Push}$(v, \Delta, F)$;\\
08.   \hspace{.30 in}	   \textbf{else} if $v \notin F$ \textbf{then do}
 \textit{Enhanced-LMES-Push/Relabel}$(v, \Delta)$;\\	
09.   	\hspace{.30 in}		\textit{FRF-Recursive-Delete-and-Merge}$(\Delta, F)$;\\
10.   	\hspace{.15 in}  \textbf{endwhile}\\
11.   	\hspace{.15 in}  \textit{Get-Next-Scaling-Parameter}$(\Delta)$;\\
12.   \textbf{end}\\

\noindent \textbf{Line 02.}   A newly violating node $v$ is a node such that $e(v) < \epsilon  \Delta$ at the beginning of the scaling phase, and such that node $v$ was special in the previous scaling phase.  \\

\noindent \textbf{Line 02.}   $\textit{FRF-Add}(v, \Delta, F)$ appends node $v$ as a leaf of $F$ 
(and possibly adds other nodes to $F$) in such a way that the path from $Root(F, v)$ to $v$ 
consists of FRF-eligible arcs.   The pseudo-code for the procedure \textit{FRF-Add} is given in 
Subsection \ref{subsec:operations}.\\

\noindent \textbf{Line 03.}   If $v$ is violating and if $FlowNeeded(v) = \frac{\Delta}{k}$, then $\frac{\Delta}{k}$ units of flow are sent from $Root(F, v)$ to $v$ along arcs in $F$.  Subsequent to sending this flow, node $v$ is no longer violating.  The pseudo-code for the procedure \textit{FRF-Pull} is given in Subsection \ref{subsec:operations}.\\

\noindent \textbf{Line 04.}  If any leaf node of $F$ is non-violating, it is deleted from $F$.  
Any bi-abundant arc with at most one endpoint in $F$ is contracted.  
The pseudo-code for the procedure \textit{FRF-Recursive-Delete-and-Merge} is given in 
Subsection \ref{subsec:pseudo}.   It relies on the procedure \textit{FRF-Delete}, 
which is described in Subsection \ref{subsec:operations}.\\

\noindent \textbf{Line 05.}  
A special node $v$ has \textit{medium excess} (resp., \textit{large excess}) if 
$e(v)\ge  \frac{\Delta}{k} - 1.5 \epsilon^4 \Delta$ (resp., $ e(v)\ge \frac{\Delta}{2} - 1.5 \epsilon^4 \Delta$).  
A normal node $v$ has \textit{medium excess} (resp., \textit{large excess}) 
if  $e(v)\ge  \frac{\Delta}{k} + \epsilon \Delta$ 
(resp.,  $e(v)\ge  \frac{\Delta}{2} + \epsilon \Delta$).\\

\noindent \textbf{Line 07.}  $\textit{FRF-Push}(v, \Delta, F)$ sends $\frac{\Delta}{k}$ units of flow in arcs of $F$ from $v$ to a violating node $w$ of $F$.  The node $w$ is a leaf node that is a descendant of node $v$.  After the push, node $w$ becomes non-violating.   The pseudo-code for the procedure \textit{FRF-Push} is given in Subsection \ref{subsec:operations}.\\

\noindent \textbf{Line 08.}   The Procedure \textit{Enhanced-LMES-Push/Relabel}$(v, \Delta)$ guarantees that the following invariant is always satisfied:   if $r_{ij} \le r_{ji}$, and if $(i, j)$ is not bi-abundant, then $r_{ij}$ is a multiple of $\frac{\Delta}{k}$.   Because this invariant is satisfied, it is also true that the residual capacity of each anti-abundant arc is a multiple of $\frac{\Delta}{k}$.  The pseudo-code for the procedure \textit{Enhanced-LMES-Push/Relabel} is given in Subsection \ref{subsec:newpush}. \\ 

\noindent \textbf{Line 11.}   Typically, the scaling parameter $\Delta$ is replaced by $\frac{\Delta}{k}$.   An exception is when there would be no push for any of the next $2Q$ scaling phases and where $F = \emptyset$.    In this latter case, the scaling parameter is chosen so that in the next scaling phase one of the following two events takes place (i) there is a push, or (ii) a node becomes violating and $F$ becomes non-empty.  The pseudo-code for the procedure \textit{Get-Next-Scaling-Parameter} is given in Subsection \ref{subsec:pseudo}.

\section{Abundant arcs and contractions}

We say that an arc $(i, j)$ is {\it abundant} in the $\Delta$-scaling phase of the Enhanced LMES if $r_{ij}\ge M\Delta = 16n^2 \Delta$ at the beginning of the scaling phase.      Abundant arcs are guaranteed to have positive residual capacity throughout the remainder of the execution of the algorithm.  A common approach for transforming a weakly polynomial flow algorithm into a strongly polynomial time algorithm involves the contraction of directed cycles of abundant arcs.

Lemma \ref{lem:mostpush} states that the total amount of flow pushed in a scaling phase of the LMES Algorithm is at most $2n^2 \Delta$.   We will show in Subsection \ref{totalflowsection} that the total amount of flow pushed in the $\Delta$-scaling phase of the enhanced LMES Algorithm is less than $5n^2\Delta$.  As a corollary, one can show that an arc that is abundant in the $\Delta$-scaling phase is also abundant in the next scaling phase, and hence in all subsequent scaling phases.  To see why, note assume that  $r_{ij}\ge M\Delta$ at the beginning of the $\Delta$-scaling phase.   At the end of the phase,  $r_{ij}\ge M\Delta - 5n^2\Delta \ge \frac{M\Delta}{2}$.   If $\Delta'$ denotes the scaling parameter at the next scaling phase, then  $\Delta' \le  \frac{\Delta}{2}$ and $r_{ij} > M \Delta'$, and  $(i, j)$ remains abundant.


%


An \textit{abundant cycle} is a directed cycle of abundant arcs.   When the algorithm identifies an abundant cycle $W$, it contracts the cycle into a single node that we refer to as a \textit{merged node}.   The contraction of $W$ reduces the number of nodes in the network by  $|W| - 1$. Since each abundant cycle has at least two nodes, the number of contractions is at most $n-1$.  The number of merged nodes created over all scaling phases is at most $n-1$.  The total number of arcs in the contracted abundant cycles is at most $2n-2$.

Each of the $n-1$ contractions can be carried out in $O(m)$ time.  Thus, the total time to carry out the contractions is $O(nm)$, which is not a bottleneck.   

After contracting an abundant cycle, the Enhanced LMES algorithm continues on the contracted network starting from where it left off.  Eventually, the algorithm identifies a maximum flow, possibly after additional contractions of abundant cycles.  Then the algorithm expands the merged nodes into abundant cycles and recovers a maximum flow in the original network.  We will explain how to expand pseudo-nodes and recover the maximum flow in Subsection \ref{sec:expansion} in the Appendix.  The expansion of the contracted cycles is not a bottleneck operation.

The contraction of abundant cycles is an approach based on the seminal work by Tardos \cite{T86} for transforming weakly polynomial time algorithms into strongly polynomial algorithms.   Orlin \cite{O13} contracted abundant cycles in his $O(nm)$ max flow algorithm.  In addition, contraction of cycles was an important aspect of Goldberg and Rao's \cite{GR98} weakly polynomial max flow algorithm. 

Of special importance to the algorithm is  the contraction of bi-abundant arcs. 
Also of importance are the reversals of abundant arcs.   If $(j, i)$ is abundant and if $(i, j)$ is not abundant, then we say that arc $(i, j)$ is {\it anti-abundant}.   We will explain how the algorithm efficiently maintains lists of abundant, anti-abundant, and bi-abundant arcs in Section \ref{sec:biabundant} of the appendix.

\subsection{Distance labels following a contraction}

In the generic push/relabel algorithm, the running time analysis relies on distance labels never decreasing.   However, when we introduce contractions, distance labels might need to decrease.   This is illustrated in Figure \ref{fig:contract}. \\

 Figure \ref{fig:contract} shows a subset of nodes before and after a contraction.   Unless $d(1)$ is decreased following the contraction (or $d(6)$ is increased), there is no way of assigning a distance label to node $w$ so that the distance labels are valid.\\

In our Enhanced LMES Algorithm, we permit distance labels to decrease subsequent to the contraction of an abundant cycle $W$.


We  assume that prior to the contraction of $W$, $d(v)$ is the minimum length of a path $v$ to $t$.  We can make the assumption true by relabeling the nodes in $O(m)$ time prior to the contraction.   

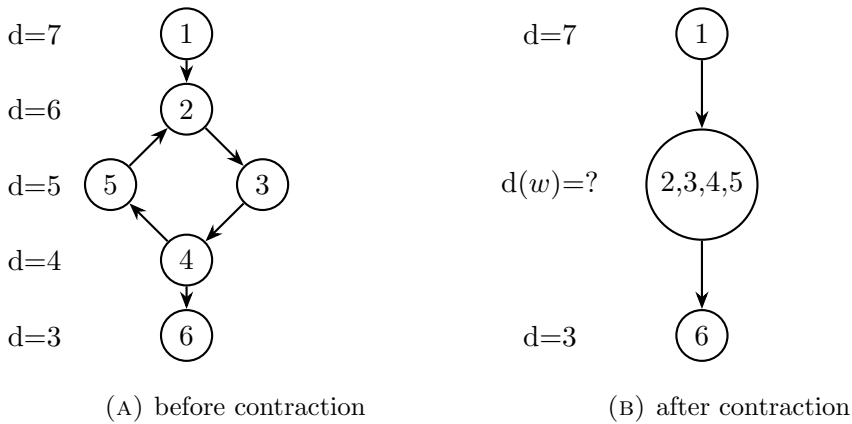
\begin{figure}[H]
\centering
\begin{subfigure}[b]{.4\textwidth}
\begin{tikzpicture}
\begin{scope}[every node/.style={circle,thick,draw}]
    \node (6) at (3,-1) {6};
    \node (5) at (2,1) {5};
    \node (2) at (3,2) {2};
    \node (4) at (3,0) {4};
    \node (1) at (3,3) {1};
    \node (3) at (4,1) {3};
\end{scope}

\begin{scope}[every node/.style={black, circle}]
    \node (k+4) at (1,3) {d=7};
    \node (k+3) at (1,2) {d=6};
    \node (k+2) at (1,1) {d=5};
    \node (k+1) at (1, 0) {d=4};
    \node (k) at (1, -1) {d=3};
\end{scope}

\begin{scope}[>={Stealth[black]},
              every node/.style={fill=white,circle},
              every edge/.style={draw=black,thick}]
    \path [->] (1) edge (2);
    \path [->] (2) edge (3);
    \path [->] (3) edge (4);
    \path [->] (4) edge (5);
    \path [->] (5) edge (2);
    \path [->] (4) edge (6);
\end{scope}
\end{tikzpicture}
\caption{before contraction}
\end{subfigure}
\begin{subfigure}[b]{.4\textwidth}
\begin{tikzpicture}
\begin{scope}[every node/.style={circle,thick,draw}]
    \node (0) at (3,3) {1};
    \node (1) at (3,1) {2,3,4,5};
    \node (8) at (3,-1) {6};
\end{scope}

\begin{scope}[every node/.style={black, circle}]
    \node (k+5) at (1, 3) {d=7};
    \node (k+1) at (1, 1) {d($w$)=?};
    \node (k) at (1, -1) {d=3};
\end{scope}

\begin{scope}[>={Stealth[black]},
              every node/.style={fill=white,circle},
              every edge/.style={draw=black,thick}]
    \path [->] (1) edge (8);
    \path [->] (0) edge (1);
\end{scope}
\end{tikzpicture}
\caption{after contraction}
\end{subfigure}
\caption{Part of a graph before and after contraction}
\label{fig:contract}
\end{figure}

After the contraction of an abundant cycle, the algorithm lets $d(s) = |N^c|$, and lets $d(v)$ be the minimum length of a path $v$ to $t$ in $G^c(x)$.  If there is no path from node $v$ to node $t$, then $d(v) = d(s) + 1$.   Then the current arc of each node is set to the first arc of its arc list, and the current node of each merged node is set to the first node in its node list. 
The time to update the distance labels is $O(m)$ per contraction and $O(nm)$ in total. 

Even though the distance label of a node $v$ may decrease following a contraction, the total number of increases in $d(v)$ is still bounded above by $n+1$.      

\begin{lem}  
\label{lem:labelscontract} 
For each node $v$, the total number of increases in $d(v)$ over all iterations of the LMES with contraction is at most $n+1$.
\end{lem}
\begin{proof}
As the algorithm proceeds, let $n'$ denote the number of nodes in the (contracted) network.  For each node $v$, let $d'(v) = d(v) - (n - n')$.    Prior to any contractions in the network, $n' = n$, and $d'(v) = d(v)$.  At every iteration, $d'(s) = n' + (n - n') = n$.

Note that $d'(v)$ increases by 1 whenever $v$ is relabeled, and $d'(v)$ does not decrease following a contraction.   Moreover, $d'(v) \le d'(s)+1 = n+1$.  Thus the number of relabels of node $v$ is at most $n+1$.  
\end{proof}

\subsection{Large excesses following a contraction}

Suppose cycle $W$ is contracted, resulting in merged node $w$.   After the contraction,  $e(w)$ is the sum of the excesses of nodes of $W$.  After the contraction, $e(w)$ can be nearly as large as 
$|W| \times \Delta$.  The running time analysis for the LMES algorithm relied on excesses not exceeding $\Delta$.   We can circumvent that difficulty as follows.   We let $N(\Delta)$ denote the nodes (including merged nodes) at the beginning of the $\Delta$-scaling phase.   The potential functions $\Phi_1$ and $\Phi_2$ include summations over the nodes in $N(\Delta)$ at the beginning of the phase.  Suppose that a merged node $w$ is obtained during the phase by contracting a cycle $W$.   Rather than replacing the nodes of $W$ by $w$ in the summation for the potential functions, we continue to sum over the nodes of $N(\Delta)$. For each node $v \in W$, we define $e(v)$ to be $\frac{e(w)}{|W|}$.  And we define $d(v)$ to be $d(w)$.   Upon contraction of the cycle $W$, the potential functions $\Phi_1$ and $\Phi_2$ do not increase.  Moreover, $e(v)$ does not exceed $\Delta$ during the phase.

For further details, see Theorem \ref{th:main}.

%

\bigskip

\section{The contraction lemma and its implications}
\label{sec:contraction}

In this section, 
we state and prove the Contraction Lemma, which provides conditions that guarantee the existence of a bi-abundant arc  within $O(\log_k n)$ scaling phases.   The Contraction Lemma is a variation of a lemma proved in Orlin \cite{O13} as part of an  $O(nm)$ max flow algorithm.  The proof here is new.  It is  simpler than the proof in \cite{O13}.

\subsection{Small,  medium, and large arcs}

Let $G(\Delta) = (N(\Delta), A(\Delta))$ denote the (contracted) network at the beginning of the 
$\Delta$-scaling phase. Recall that an arc is called \textit{small},  \textit{medium}, or \textit{large} according as 
 $u_{ij} + u_{ji} <  \epsilon^5 \Delta$ or $\epsilon^5 \Delta \le u_{ij} + u_{ji} < 2M \Delta$ or    
$ u_{ij} + u_{ji} \ge 2M \Delta$. 
Let Small$(\Delta)$, Medium($\Delta$), and Large$(\Delta)$ denote the sets of small, medium, and large arcs of the
$\Delta$-scaling phase.

If $(i, j)$ is large at the $\Delta$-scaling phase, then $(i, j)$ must be abundant or anti-abundant. Let  Abundant$(\Delta)$ and Anti$(\Delta)$ denote the subset of arcs of Large($\Delta)$  that are abundant and anti-abundant arcs at the beginning of the $\Delta$-scaling phase.  The following two relationships follow directly from the definitions.

\begin{itemize}
\label{lem:small}
\item  $A(\Delta) = \text{Small}(\Delta) \cup \text{Medium}(\Delta) \cup \text{Large}(\Delta)$.
\item   $\text{Large}(\Delta) = \text{Abundant}(\Delta) \cup \text{Anti}(\Delta)$.
\end{itemize}

We say that there is an \textit{occurrence} of an arc $(i, j)$ (resp., node $v$) at the $\Delta$-scaling phase if $(i, j) \in A(\Delta)$ (resp., $v \in N(\Delta))$.    

Each arc in $A$ can remain medium for at most $O(\log_k n)$ scaling phases.  We express this fact in terms of the total number of occurrences of medium arcs.

\begin{lem}
\label{lem:mediumarcs}
The total number of occurrences of medium arcs over all scaling phases is 
$O(m \log_k n)$.  
\end{lem}

\subsection{Node imbalances and the contraction lemma}

For any preflow $x$ defined at or after the $\Delta$-scaling phase, the 
{\it $\Delta$-imbalance} of node $v$ with respect to $x$ is:   

\begin{align*}
\text {IMB}(v, x, \Delta) = e_x(v) + \sum_{(j, v) \in Anti(\Delta)} r_{j v}(x)
 - \sum_{(v, j) \in Anti(\Delta)} r_{vj}(x).
\end{align*}

We let IMB$(v, \Delta)$ refer to IMB$(v,x, \Delta)$ where $x$ is the initial preflow of the 
$\Delta$-scaling phase.  

 If $y$ and $x$ are preflows, then $y - x$ denotes the difference. 
Our first result is a simplification of the formula for IMB$(v, y-x, \Delta)$ in the case that $x$ and $y$ are both preflows obtained at or after the $\Delta$-scaling phase.

\begin{lem}
\label{lem:imbeq}
Suppose that preflow $x$ is obtained by the algorithm at the $\Delta$-scaling phase.  Suppose  that $y$ is a preflow obtained at the $\Delta'$-scaling phase for $\Delta' \le \Delta$.  Suppose further that node $v$ is not incident to a bi-abundant arcs at either of these two scaling phases.   Then 

\begin{align}
\text {IMB}(v, y-x, \Delta) = e_{y-x}(v) 
 - \sum_{(v, j) \in Large(\Delta)} r_{vj}(y-x).
\label{eq:imb1}
\end{align}

\noindent In addition, if node $v$ is not incident to any medium arcs at the $\Delta$-scaling phase, then the following three equalities and inequality are valid.

\begin{align} 
e_{y-x}(v) 
= \sum_{(v, j) \in \textit{Small}(\Delta) \cup  \textit{Large}(\Delta)} r_{vj}(y-x).
\label{eq:imb3}
\end{align}

\begin{align}
\text {IMB}(v, y-x, \Delta) &=  \sum_{(v, j) \in Small(\Delta)} r_{vj}(y-x)  
\label{eq:imb4}
 \end{align}
 
\begin{align}
 \label{eq:imb2}
|\text {IMB}(v, y-x, \Delta)| &\le   
 \sum_{(v, j) \in Small(\Delta)} (u_{vj} + u_{jv}) \hspace{10 pt} \le 2n \epsilon^5 \Delta < .5 \epsilon^4 \Delta.
 \end{align}
\end{lem}

\begin{proof}

For each arc $(i, j) \in A$, $r_{ij}(y-x) = -r_{ji}(y-x)$.   By applying this observation to arcs $(j, v) \in$ Anti$(\Delta)$, one obtains Equation \ref{eq:imb1}.

We next prove that Equation \ref{eq:imb3} is valid.  If there are no medium 
arcs incident to node $v$ at the $\Delta$-scaling phase, then $A(\Delta) = \textit{Small}(\Delta) \cup \textit{Large}(\Delta)$.  In sending $\delta$ 
units of flow in any arc $(i, j)$, $r_{ij}$ decreases by $\delta$, and 
$r_{ji}$ increases by $\delta$.   To keep track of changes in flow incident to node $v$, it suffices to 
consider the change in residual capacity of arcs leaving node $v$.  Therefore,  Equation \ref{eq:imb3} is valid. 
Equality \ref{eq:imb4} follows from equations  \ref{eq:imb1} and  \ref{eq:imb3}.
Inequality \ref{eq:imb2} follows directly from equation \ref{eq:imb4}.
\end{proof}

We say that node $i$ is \textit{$\alpha$-balanced} in the $\Delta$-scaling phase if 
$|\text{IMB}(v, \Delta)| \le \alpha$.  We say that $i$ is \textit{$\alpha$-imbalanced}  if 
$|\text{IMB}(v, \Delta)| > \alpha$. 


\begin{lem} 
\label{lem:contraction}
{\bf (Contraction Lemma)}.   Suppose that node $v$ is not incident to any medium arcs at the $\Delta$-scaling phase and that $v$ is $(\epsilon^4\Delta)$-imbalanced at the beginning of the $\Delta$-scaling phase.    Then node $v$ will become incident to a bi-abundant arc within $7Q$ scaling phases. 
\end{lem} 

\begin{proof}

Let $x$ denote the preflow at the beginning of the $\Delta$-scaling phase.  Thus 
$|\textit{IMB}(v, x, \Delta)| > \epsilon^4\Delta$.   Let $\Delta' =  \epsilon^7\Delta$ be the 
scaling parameter $7Q$ scaling phases subsequently.  
Let $y$ be the preflow at the beginning of the $\Delta'$-scaling phase.  Then

\begin{align*}
|\text {IMB}(v, y, \Delta)| & \ge  |\text {IMB}(v, x, \Delta)| - |\text {IMB}(v, y-x, \Delta)|   \\ 
&\ge  \epsilon^4\Delta - 2n  \epsilon^5\Delta \ge .5  \epsilon^4\Delta. 
 \end{align*}

\noindent The second inequality follows from Lemma \ref{lem:imbeq}.    Note that $e_y(v) < \Delta'$.   We now consider two cases according as $\text {IMB}(v, y, \Delta)$ is positive or negative.
If $\text {IMB}(v, y, \Delta) > 0$, then 

\begin{align*}
\text {IMB}(v, y, \Delta)  = e_y(v) + \sum_{(j, v) \in Anti(\Delta)} r_{j v}(y)
 - \sum_{(v, j) \in Anti(\Delta)} r_{vj}(y) >  .5  \epsilon^4\Delta > nM \Delta'.
\end{align*}

In this case, there is an arc $(j, v) \in Anti(\Delta)$ such that $r_{j v}(y) > M\Delta'$, in which case arc $ (j, v) $ is bi-abundant at the $\Delta'$-scaling phase.  
If, instead,  $\text {IMB}(v, y, \Delta) < 0$, then
 
 \begin{align*}
\text {IMB}(v, y, \Delta)  = e_y(v) + \sum_{(j, v) \in Anti(\Delta)} r_{j v}(y)
 - \sum_{(v, j) \in Anti(\Delta)} r_{vj}(y) <  -.5  \epsilon^4\Delta < -nM \Delta'.
\end{align*}

In this case, there is an arc $(v, j) \in Anti(\Delta)$ such that $r_{vj}(y) > M\Delta'$, in which case arc $ (v, j) $ is bi-abundant at the $\Delta'$-scaling phase.   This completes the proof.
\end{proof}

\bigskip

\section{Special, Violating, and Normal nodes.}


In this section, we review the three types of nodes: special nodes, violating nodes, and normal nodes.   We describe properties and invariants satisfied by these nodes.   We also provide the procedures \textit{Initialize}, \textit{Select}, and \textit{Push} for the Enhanced LMES Algorithm.

\subsection{Three flavors of nodes}  The Enhanced LMES Algorithm partitions the nodes of $N(\Delta) \backslash \{s, t\}$ into three different flavors (types).  For each node $v$, we also define the modified excess $\hat e(v, \Delta)$, which replaces $e(\,)$ in the procedures within the Enhanced LMES Algorithm.  \\




\noindent {\bf  Nodes that are special at the $\Delta$-scaling phase}. 

A node $v$ is a {\it special node} 
during the $\Delta$-scaling phase if it satisfies the following.  
\begin{enumerate} 
\item  Node $v$ is not incident to any medium arcs, and 
\item  $|\text{IMB}(v, \Delta)| \le  \epsilon^4 \Delta$.
\end{enumerate} 

For each special node $v \in N(\Delta)$, we define $\hat e(v, \Delta) = e(v) + 1.5 \epsilon^4 \Delta$.   The following are properties of nodes that are special at the $\Delta$-scaling phase.

\begin{itemize}
\item [S1.]   $\text{IMB}(v, \Delta ) \le \epsilon ^4 \Delta$.  (By definition.)
\item [S2.]   $v$ is not incident to any medium arcs.  (By definition.)
\item [S3.]    If an anti-abundant arc $a$ is incident to $v$, then for each preflow $x$ during the $\Delta$-scaling phase, $r_a(x) \equiv 0 \left( \kern -8pt \mod  \frac{\Delta}{k} \right)$.  (See Lemma \ref{lem:emod}.)
\item [S4.]   $e(v) \ge -1.5 \epsilon^4 \Delta$, and thus $ \hat e(v) \ge 0$.  (See Lemma \ref {lem:emod}.)
\item [S5.]    If $\hat e(v) > 3 \epsilon ^4$, then $\hat e(v) \ge \Delta /k$. Equivalently, if $\hat e(v) > 3 \epsilon ^4$, then $v$ will either have medium or large excess. (See Lemma \ref {lem:emod}.)
\end{itemize}
\bigskip


\noindent {\bf  Nodes that are normal at the $\Delta$-scaling phase}.  

We will ``require'' every non-special node to maintain an excess of at least $\epsilon \Delta $ during the $\Delta$-scaling phase.  This small amount of buffer is needed later when we add nodes to the flow-return forest.

 A node $v$ is called a {\it normal node} 
during the $\Delta$-scaling phase if $e(v) \ge \epsilon \Delta$.  


For each normal node $v \in N(\Delta)$, we define $\hat e(v, \Delta) = e(v) -  \epsilon \Delta$. 
The following are properties of nodes that are normal during the $\Delta$-scaling phase.

\begin{enumerate}
\item [N1.]   If node $v$ is normal during the $\Delta$-scaling phase, then $\hat e(v, \Delta) \ge 0$.  (By definition.)
\item [N2.]     If node $v$ is normal during the $\Delta$-scaling phase, then $\hat e(v, \Delta) \le \Delta + (k-1) \epsilon \Delta $.  (See Lemma \ref {lem:ehat}.)   
\item [N3.]  The total number of occurrences of normal nodes is 
$ O(m \log_k n)$.  (See Lemma \ref {lem:normal}.)
\end{enumerate}

\smallskip

\noindent {\bf  Nodes that are violating at the $\Delta$-scaling phase}.

A node $v$ is a {\it violating node} 
at the $\Delta$-scaling phase if $e(v) < \epsilon \Delta$ and $v$ is not special.
For each violating node $v \in N(\Delta)$, we define $\hat e(v, \Delta) = e(v) -  \epsilon \Delta$. 
A node is called \textit{newly violating} at the $\Delta$-scaling phase if it is violating at the $\Delta$-scaling phase but not violating at the previous phase.  

The following are properties of nodes that are violating at the $\Delta$-scaling phase.  Several of these properties mention a data structure called the flow-return forest, which is described in Section \ref{sec:FRF}. 

\begin{enumerate}

\item [V1.]   $ \hat e (v, \Delta ) < 0 $.  (By definition.)  
\item [V2.]       If $v$ is newly violating at the $\Delta$-scaling phase, then $v$ is added to the flow-return forest at that phase.   (Refer to the Procedure \textit{FRF-Add} in Section \ref{sec:FRF}.)
\item [V3.]      If $v$ is violating at the $\Delta$-scaling phase  but not newly violating, then $v$ is already in the flow-return forest at the beginning of the $\Delta$-scaling phase.  Node $v$ will be deleted from the flow-return forest at the iteration at which $v$ is non-violating and $v$ is a leaf of the flow-return forest. 
\item [V4.]      If $v$ is violating at the $\Delta$-scaling phase, then $v$ will become non-violating within $2Q$ scaling phases. (See Lemma \ref{lem:2Q}.)
\item [V5.]      If $v$ is violating, then $\hat e(v) \ge  - \epsilon  \Delta$.   (Corollary \ref{cor:2Q}.)
\item [V6.]     The total number of occurrences of newly violating nodes is $ O(m \log_k n)$.  (See Lemma \ref {lem:normal}.)

\end{enumerate}

\bigskip

We have defined $\hat e(v, \Delta)$ for special, violating and normal nodes.   In Section \ref{sec:FRF}, we will modify our definition of $\hat e(v, \Delta)$ for root nodes of the flow-return forest.


\subsection{A new flow invariant}
\label{subsec:newinvariant}

In order to ensure that Property S3 is satisfied for special nodes, we will require that an even stronger (more restrictive) invariant is satisfied.

\begin{inv}
If $r_{ij} <  r_{ji}$, then $r_{ij} \equiv 0 \left (\kern -8pt \mod \frac{\Delta}{k} \right)$.
\label{inv:2B}
\end{inv}

We will adjust our rules for pushing later in this section.  The modified rules ensure that Invariant \ref{inv:2B} is satisfied at all iterations.   
If $(i, j)$ is an anti-abundant arc, then $r_{ij} <  r_{ji}$.   By Invariant \ref{inv:2B},  $r_{ij} \equiv 0 \left (\kern -8pt \mod \frac{\Delta}{k} \right)$.  In addition, the following lemma is true.

%
%
%

\begin{lem}
\label{lem:emod}
Suppose that node $v$ is special throughout the $\Delta$-scaling phase.  
Then for each preflow $y$ during the phase,  
$0 \le \text{Mod} \left(\hat e_y(v), \frac{\Delta}{k} \right) \le 3 \epsilon^4 \Delta$.  
\end{lem}

\begin{proof}
Let $x$ denote the initial preflow of the phase, and let $y$ be any other preflow during the phase.  
%
%
Then 

\begin{enumerate}
\item  $\hat e_y(v) = e_y(v) + 1.5 \epsilon^4 \Delta$.
\item  $e_y(v) = e_x(v) - e_{x-y}(v)$.
\item $e_x(v) \equiv \text{IMB}(v, \Delta)  
	( \kern -8pt \mod \frac{\Delta}{k})$.
\item $  -\epsilon^4 \Delta \le \text{IMB}(v, \Delta)  \le \epsilon^4 \Delta$ (because 		$v$ is special). 
\end{enumerate}
\smallskip

Adding (1) to (4) and combining terms yields the following:
$$  \epsilon^4 \Delta  \le \text{Mod} \left(\hat e_y(v) + e_{x-y}(v) + .5 \epsilon^4 \Delta , \frac{\Delta}{k} \right) \le 3 \epsilon^4 \Delta.$$.

Lemma \ref{lem:imbeq} implies that 
$$ 0  \le \text{Mod} \left( e_{x-y}(v) + .5 \epsilon^4 \Delta , \frac{\Delta}{k} \right) \le \epsilon^4 \Delta.$$

Combining the above  inequalities completes the proof.
%
%
%
%
%
%
\end{proof}

%

Using the above notation, we observe the following, which we state as a lemma.

\begin{lem}
\label{lem:violatingnode}
Suppose that v is normal in the $\Delta$-scaling phase as well as in the previous scaling phase.   Then $\hat e(v, \Delta) \ge  (k-1) \epsilon \Delta $  at the beginning of the $\Delta$-scaling phase. 
\end{lem}

\begin{proof}
At the end of the $k\Delta$-scaling phase and also at the beginning of the $\Delta$-scaling phase, $e(v) \ge \epsilon (k\Delta)$.   Therefore, at the beginning of the $\Delta$-scaling phase, $\hat e(v, \Delta) = e(v, \Delta) - \epsilon \Delta \ge (k-1) \epsilon \Delta$.
\end{proof}

\smallskip

\subsection{Initialization of the Enhanced LMES Algorithm} 
\label{subsec:initialization}

The Initialization procedure of the Enhanced LMES Algorithm finds an initial preflow satisfying the following.  

\begin{enumerate}  
\item  All arcs in $A^+(s)$ are saturated.
\item   $d(s) = n$; for $i \neq s$, $d(i) = 0$.
\item   For any arc $(i, j)$ with $r_{ij} < r_{ji}$, $r_{ij}\equiv 0  \left( \kern -8pt  \mod  \frac{\Delta}{k} \right)$. 
\item   For all $v \in N$, $\hat e(v) \ge 0$. 
\end{enumerate}

\bigskip

\noindent \textbf{\textit{Initialize-Enhanced-LMES}}\\
01.  \textbf{begin}\\
02.  \hspace{.15 in}   \textbf{for each } arc $(i, j) \in A$, $x_{ij} := 0$;\\
03.  \hspace{.15 in}  \textbf{for each } node $i \in  N$, $d(i) := 0$;\\
04. \hspace{.15 in} \textbf{for each } arc $(s, i) \in A^+(s)$  send $ r_{si}$ units of flow in $(s, i)$;\\
05.  \hspace{.15 in} $d(s) := n$;\\
06. \hspace{.15 in}  $\Delta := \max\{e(i) : i \in  N\}$;\\
07. \hspace{.15 in}    \textbf{for each } arc $(i, j)$ with $ r_{ij} < r_{ji}$, send $\text{Mod} \kern -3pt \left(r_{ij},  
\frac{\Delta}{k} \right)$ units of flow in $(i, j)$;\\
08.  \hspace{.15 in}   \textbf{for each } node $i \in N \backslash \{s,t \}$ with $e(i) < \epsilon \Delta$, 
send $\frac{ \Delta}{k}$ units of flow in arc $(t, i)$;\\
09.   \textbf{end}  \\

Steps 07 and 08 guarantee that invariant 2 is satisfied.
Step 08 guarantees that each node $v \neq s \text{ or } t$ is normal and $\hat e(v) \ge 0$.  Step $08$ accomplishes this by permitting flows in the arcs directed out of node $t$.

\bigskip


\subsection{Large and medium excess nodes, and a new push procedure}
\label{subsec:newpush}

The Enhanced LMES Algorithm bases the definition of large and medium excess nodes using $\hat e$ rather than $e$. 
Node $v$ has {\it large excess} at the $\Delta$-scaling phase if 
$\hat e(v, \Delta) \ge \frac{\Delta}{2}$.  
Node $v$ has  {\it medium excess} at the $\Delta$-scaling phase if 
$\frac{\Delta}{k} \le  \hat e(v, \Delta) < \frac{\Delta}{2}$.

During the $\Delta$-scaling phase, LargeSet and MediumSet are the sets of large and medium excess nodes.  
To identify MediumSet and LargeSet, one needs to evaluate $\hat e$, which requires that one needs to know whether a node is special, which requires that one can identify when arcs are anti-abundant.  We discuss efficient implementation of these terms in Subsection \ref{sec:biabundant} in the appendix.

The rule for selecting active nodes is the same for the Enhanced LMES Algorithm as for the LMES Algorithm with the exception that LargeSet and MediumSet are defined in terms of $\hat e$ rather than $e$.\\

\noindent  \textbf{\textit{Enhanced-LMES-Select-Node}} \\ 
\indent  Applicability:  there is a node with medium or large excess.\\
\indent  Action: if LargeSet $\neq \emptyset$, then let 
$i = \text{argmin} \{d(j) : j \in \text{LargeSet} \}$; \\
\indent\indent   else,  $i =$ argmax $\{d(j) : j \in$ MediumSet\}.\\

As before, the phase ends when there are no medium or large excess nodes.

The procedure for pushing flow in an arc is given below.   It ensures that  Invariant \ref{inv:2B} is always satisfied.  The procedure is called when node $i$ is selected for pushing and node $i$ is not  a node of 
the flow-return forest.  If $i$ is a node of the flow-return forest, then the procedure \textit{FRF-Push}$(i)$ is called instead, as described in Section \ref{sec:FRF}. \\

\noindent \textbf{\textit{Enhanced-LMES-Push$(i, j, \Delta)$}}\\
Applicability: node $i$ is active, arc $(i,j)$ is admissible, and node $i$ is not a node of the flow-return forest.\\
Action:   $\delta$ units of flow is sent in arc $(i, j)$ so as to satisfy Invariant  \ref{inv:2B}. \\  

\noindent 01.  \textbf{begin}\\
02.	\hspace{.15 in}  \textbf{if } $\hat e(i) \ge \frac{\Delta}{2}$,   \textbf{then } $D := \frac{\Delta}{2}$;  \\
03.	\hspace{.30 in} \textbf{else} 
	$D := \hat e(i) - \text{Mod} \left(\hat e(i), \frac{\Delta}{k}\right)$; \\
04.	\hspace{.15 in}  \textbf{if } $r_{ij} < D$,  \textbf{then } $\delta := r_{ij}$;\\
05.	\hspace{.15 in}  \textbf{else if } $r_{ij} - r_{ji} \ge 2D$,  \textbf{then } $\delta := D$;\\
06.	\hspace{.15 in}  \textbf{else if } $r_{ij} < r_{ji}$,  \textbf{then } $\delta := D$;\\
07.	\hspace{.15 in}  \textbf{else if } $ 0 \le  r_{ij} - r_{ji} < 2D$,  \textbf{then } $\delta := \frac{(r_{ij} - r_{ji})}{2}$;\\
08.	 \hspace{.15 in} \textbf{else if }  $r_{ij} = r_{ji}$,  \textbf{then } $\delta := D - \text{Mod}\left(r_{ij}, \frac{\Delta}{k} \right)$ \\
09.	 \hspace{.15 in}  send $\delta$ units of flow in arc $(i, j)$;  \\
10.  \textbf{end}\\

In Steps 02 and 03, $D$ is chosen to be the largest integral multiple of $\frac{\Delta}{k}$ that is at most $\min \{\hat e(i), \frac{\Delta}{2}\}$.

Suppose that Invariant \ref{inv:2B} is true prior to running \textit{Enhanced-LMES-Push$(i, j, \Delta)$}. We claim that it will remain true subsequent to the push.   
Let $r'_{ij}$ denote the residual capacity before the push.  Let $r_{ij}$ denote the residual capacity of $(i, j)$ after the push.   
\begin{enumerate}
\item  If $\delta$ is chosen as in Step 04, then $r_{ij} = 0$.
\item   If $\delta$ is chosen as in Step 05, then 
$r'_{ji} < r_{ji} < r_{ij} < r'_{ij}$.  In this case,  $r'_{ji}  \equiv 0 \left(\kern -8pt \mod \frac{\Delta}{k} \right)$, and
$r_{ji}  \equiv 0 \left(\kern -8pt \mod \frac{\Delta}{k} \right)$.
\item  If $\delta$ is chosen as in Step 06, then 
$r'_{ij}  \equiv 0 \left(\kern -8pt \mod \frac{\Delta}{k} \right)$, and
$r_{ij}  \equiv 0 \left(\kern -8pt \mod \frac{\Delta}{k} \right)$.
\item If $\delta$ is chosen as in Step 07, then $r_{ij} = r_{ji}$.
\item  If $\delta$ is chosen as in Step 08, then 
$r_{ij} < r'_{ij} = r'_{ji} < r_{ji}$  and $r_{ij}  \equiv 0 \left(\kern -8pt \mod \frac{\Delta}{k} \right)$.
\end{enumerate}


A  push of $\delta$ units is called \textit{large} at the $\Delta$-scaling phase if $\delta = 
\frac{\Delta}{2}$.   The push is called \textit{medium} if it is not large, and $\delta = D$; that is, $\delta$ is selected in line 05 or line 06.   In this case,  
$\delta$ is chosen the largest integral multiple of $\frac{\Delta}{k}$ that is at most $\hat e(i)$.
Subsequent to the push, $\hat e(i) < \frac{\Delta}{k}$.

If $\delta$ is chosen in line 04 or 07 or 08, then the push is neither large nor medium.  For fixed arc $(i, j)$ and for fixed value of $d(i)$, each of these choices for $\delta$  can occur at most once.   Since each node is relabeled at most $n$ times, the following lemma is true.      

\begin{lem}
\label{lem:notmedium}
The number of  pushes in the Enhanced LMES Algorithm that are neither large nor medium is $O(nm)$. 
\end{lem}


\section{The flow-return forest}
\label{sec:FRF}


When a node $v$ becomes newly violating, then $\hat e(v, \Delta)$ becomes negative, and node $v$ is added to a new data structure that we call the {\it flow-return forest} (FRF).  Within $2Q$ additional scaling phases, flow is sent to node $v$, after which it is no longer violating.  In this section, we describe the FRF.  But first, we bound the number of occurrences of nodes that become newly violating.

\begin{lem}  
\label{lem:newlyviolating}
The total number of occurrences of newly violating nodes is $O(m)$.
\end{lem}

\begin{proof}   A node $v$ can become violating if $v$ is a merged node that is newly created by contracting an abundant cycle, or $v$ stops being special.  The number of occurrences of newly merged nodes is less than $n$.   We now consider  a violating node $v$ that was special at the previous scaling phase.   Either  (1) an arc incident to $v$ becomes newly medium, or (2) $ |\text{IMB}(v, \Delta)| > \epsilon^4 \Delta$.    There are at most $2m$ nodes that can become newly violating because of newly medium arcs.   (Each arc becomes medium at most once.  Each newly medium arc might cause both endpoints to become non-special.)     

Suppose instead that $v$ becomes newly violating at the $\Delta$-scaling phase because of (2) and assume that $v$ is not incident to any medium arcs.  Then 
$|\text{IMB}(v, \Delta)| > \epsilon^4 \Delta$.  We claim that $v$ cannot become special again.  To see why, let  $\Delta'$ be a scaling parameter at some phase after the $\Delta$-scaling phase.  Let $y$ denote the preflow at the beginning of the $\Delta'$-scaling phase.
  By Lemma \ref{lem:imbeq}, 
$|\text{IMB}(v, y, \Delta) - \text{IMB}(v, \Delta)| < .5\epsilon^4 \Delta$.   It follows that   
$\text{IMB}(v, y, \Delta) > .5 \epsilon^4 \Delta \ge 2 \epsilon^4 \Delta'$, and node $v$ is not special at the $\Delta'$-scaling phase.   

Thus, the total number of times that (2) can occur without (1) is less than $n$.   This completes the proof.    
 \end{proof} 

\subsection{Nodes and arcs of the FRF}

Let $F$ denote the flow-return forest.  For each node $v \in F$, let $Root(v, F)$ denote the root of node $v$.      We let \textit{Roots$(F)$} denote the set of root nodes of $F$. 

An arc $(i, j)$ is called \textit{FRF-eligible} if  $(i, j)$ is abundant and  $d(j) \le  d(i) +1$.  We require that every arc of the flow-return forest is FRF-eligible.  We also require that for each node $v \in F$, the path from $Root(v, F)$  to node $v$ is an a directed path.   Since each arc of $F$ is FRF-eligible, the path from $Root(v, F)$  to node $v$ is abundant, and the reversal of the path satisfies the distance validity conditions. 
 
In addition, every leaf node of $F$ is a violating node.   When a leaf node of $F$ becomes non-violating, it is deleted from $F$.

%

For each node $i \in F$, let $Desc(i, F)$ denote the descendants of $i$ in $F$. 

\bigskip

\subsection{NeededFlow( ) and Reserve( )}


The algorithm maintains an array called $NeededFlow(\,)$ and another array called $Reserve(\,)$.   When a newly violating node $v$ is added to $F$ at the $\Delta$-scaling phase,  $NeededFlow(v)$ is set to $ \frac{\epsilon^2\Delta}{k}$.  For each non-violating node $i \in  F$,  $NeededFlow(i) = 0$.   After $2Q$ additional scaling phases, the scaling parameter is $ \Delta' =  \epsilon^2 \Delta $.   If node $v$ is still violating at the $\Delta'$-scaling phase, then  $\frac{\Delta'}{k} $ units of flow are  sent on the path in $F$ from $Root(v, F)$ to node $v$.  Subsequently, $v$ is non-violating.

In order to ensure that $Root(v, F)$ has sufficient excess, the algorithm maintains an array $Reserve(\,)$ that satisfies the following two conditions.  For each node  $i \in Roots(F)$:

\begin{align}
\label{reserve}
Reserve(i) = \sum_{v \in Desc(i, F)} NeededFlow(v).\\
\text{If } i \in Roots(F) \text{, then } e(i) \ge Reserve(i) + \epsilon \Delta. 
\label{rootehat} 
\end{align}

\smallskip

Invariant \ref{reserve} is easily maintained as follows:  if node $v$ is added to $F$ at the $\Delta$-scaling phase then, NeededFlow$(v) :=  \frac{\epsilon^2\Delta}{k}$, and $Reserve(Root(v, F))$ is incremented by $ \frac{\epsilon^2\Delta}{k}$.  Immediately prior to deleting a node $v$ from $F$, $Reserve(Root(v, F))$ is decremented by $NeededFlow(v)$, after which $NeededFlow(v)$ is set to 0. 

We modify our definition of of $\hat e$ for nodes of $Roots(F)$ as follows:   $\text{If } i \in Roots(F) \text{, then } \hat e(i) = e(i) - Reserve(i) - \epsilon \Delta$.   Then \ref{rootehat} is equivalent to:

\begin{align}
\text{If } i \in Roots(F) \text{, then } \hat e(i)  \ge 0.  
\end{align}

We still consider all root nodes to be normal.

 For every newly violating node $v$ for which $i = Root(v, F)$,  $Reserve(i)$ is incremented by 
$\frac{\epsilon^2 \Delta} {k}$, which causes $\hat e(i)$ to be decremented by the same amount.  The maximum increase in reserve for node $i$ (and corresponding decrease of  $\hat e(i)$) at the beginning of the $\Delta$-scaling phase is $\frac{n\epsilon^2 \Delta} {k}$, which is less than $\epsilon \Delta$.   
By  Lemma \ref{lem:violatingnode}, $\hat e(i) \ge (k-1) \epsilon \Delta$ at the 
beginning of the scaling phase.   Therefore, $\hat e(i) \ge (k-2)\epsilon \Delta$ after all  additions to $F$ at the $\Delta$-scaling phase.   

In the next subsection, we prove that it is always possible to add a newly violating node to $F$ by appending a path of FRF-eligible arcs to $F$. 

\smallskip

\subsection{Adding nodes to the FRF}
For each newly violating node $v$ in the $\Delta$-scaling phase, the algorithm determines an FRF-eligible path $W$ from $w$ to $v$ where $w = t$ or $w$ is normal.   Then $F$ is replaced by $F \cup W$.
The following lemma guarantees  that such a path $W$ exists.

\begin{lem}  
\label{lem:FRFeligible}
\textbf{(The FRF-eligible Path Lemma.)}  Suppose that $v$ is a newly violating node at the start of the $\Delta$-scaling phase. Then there is a directed path  of FRF-eligible arcs from some node $w$ to node $v$, where  $w = t$ or else $w$ is a normal node.
\end{lem}

\begin{proof}  
Let $S = \{i \in N(\Delta): \text{there is an FRF-eligible path from } i \text{ to } v \}$.  We want to show that $t \in S$ or there is a normal node in $S$.  We suppose that this is not the case, and we will derive a contradiction.  Since $S$ contains no normal nodes, it contains no root nodes of $F$. Accordingly, $S$ cannot contain any node of $F$. 
It follows that $S$ consists of special nodes and newly violating nodes.  (Any violating node that is not newly violating is in $F$.)   Thus, for every node $i \in S$, $|e(i)| < \epsilon^3 \Delta$.
   
In the following, we say that a push in $(i, j)$ is \textit{special} at the 
$\Delta'$-scaling phase if node $i$ is special at the $\Delta'$-scaling phase  and if $ e(i) < \epsilon \Delta'$ after the push.  

Let $(i_1, i_2)$ be the most recent push from a node of $S$, and suppose that the push took place in the 
$\Delta'$-scaling phase.     Let $\bar e(\,)$ denote the vector of excesses immediately subsequent to the push, and let $e(\,)$ denote the vector of excesses at the beginning of the $\Delta$-scaling phase.

The push in $(i_1, i_2)$  was special since 
   $\bar e(i_1, \Delta') \le  e(i_1, \Delta) < \epsilon^3 \Delta <   \epsilon^3 \Delta'$.  Moreover, the amount of flow pushed in $(i_1, i_2)$ was at least $ \frac{\Delta'}{k}$.   After the push, node $i_2$ became a medium or large excess node.  (This is true regardless of whether $i_2$ was special or normal.) Therefore, there was a subsequent push from $i_2$ at the $\Delta'$-scaling phase.  

Since the push in $(i_1, i_2)$ was special,  arc $(i_1, i_2)$ was large during the $\Delta'$-scaling phase.   (It could not have been small, and there were no medium arcs incident to $i_1$.)\\

The following are all true:

\begin{enumerate}
\item $|  \bar e(i_1) | < 1.5 \epsilon^4 \Delta'$.
\item  $\bar e(i_2) > \frac{\Delta'}{k} - \epsilon \Delta'$.  
\item   There was no medium or large push into node $i_1$ subsequent to the push in $(i_1, i_2)$.   (There was no subsequent push from node $i_1$ and $e(i_1) < \epsilon \Delta$.)
\item  Node $i_2 \notin S$.  (The last push from a node in $S$ was in $(i_1, i_2)$, and there was a subsequent push from $i_2$.)
\item Arc $(i_1, i_2)$ is abundant.  (Otherwise, arc  $(i_1, i_2)$ is anti-abundant, which implies that arc  $(i_2, i_1)$ is FRF-eligible, which implies that $i_2 \in S$.) 
\end{enumerate} 

We now consider the sum of the imbalances of nodes of $S$.

\begin{align*}
\text{IMB}(S, \Delta) = \sum_{i \in S} \text{IMB}(i, \Delta) =
\sum_{i\in S} e(S)+\sum_{(i, w)\in Anti(\Delta) : i \in S, w \notin S} r_{iw}  -\sum_{(w, i)\in Anti(\Delta) : i \in S, w \notin S} r_{wi}.
\end{align*}

\smallskip

Note that $ |\text{IMB}(i, \Delta )| < \epsilon^3 \Delta$ for each $i \in S$, and $|e(i)| < 1.5 \epsilon \Delta$ for each $i \in S$.    Therefore, 
$ |\text{IMB}(S, \Delta) - e(S)| < 2|S| \epsilon \Delta \le .5 \Delta $.   

\smallskip

By (5) above, $(i_2, i_1)$ is anti-abundant and directed into $S$.  There was no  push in $(i_2, i_1)$ after the push in $(i_1, i_2)$.  Therefore,   
$  r_{i_2,i_1} \ge \frac{\Delta'}{k} \ge \Delta$.  Since $ |\text{IMB}(S, \Delta) - e(S)| <  .5 \Delta $, there must be some anti-abundant arc $(i_3, i_4)$ directed out of $S$ with  
$r_{i_3,i_4} > \frac{\Delta}{2m}$.  It follows that $(i_4, i_3)$ is FRF-eligible.  But this also implies that  $i_4 \in S$, contradicting that $(i_3, i_4)$ is directed out of $S$.  In all cases we have derived a contradiction.   Thus $S$ contains node $t$ or a normal node. This completes the proof.
\end{proof}

\bigskip

\subsection{Properties of nodes of the flow-return forest}
The following are properties of the nodes of the flow-return forest.  After each property, we indicate why it is true or where we will show it to be true.\\

\begin{enumerate}
  \item [FRF 1.]      Every leaf node of $F$ is violating.   (Non-violating leaf nodes are deleted from $F$.)
\item [FRF 2.]      Every node of $Roots(F)\backslash\{t\}$ is normal.  (The procedure \textit{FRF-Add} only adds normal root nodes.)
\item [FRF 3.]       If $v$ is added to F at the $\Delta$-scaling phase, then $FlowNeeded(v) = \frac{\epsilon ^2 \Delta}{k}$.  (Specified in the procedure \textit{FRF-Add}.)
\item [FRF 4.]  	  If $i \in Roots(F)$, then 
$Reserve(i) = \sum_{v\in Desc(i, F)} FlowNeeded(v)$.  (The equality is preserved when running the procedure \textit{FRF-Add}.)
\item [FRF 5.]      If $i \in Roots(F)$ at the $\Delta$-scaling phase, then 
$\hat e(i) = e(i) - \epsilon  \Delta  - Reserve(i) \ge 0$. 
(See Lemma \ref{lem:FRFeligible}.)
\item [FRF 6.]    For each normal node $i \in F$ , 
$\hat e(i) \le \Delta + (k-1) \epsilon \Delta$.  (See Lemma \ref {lem:ehat}.)
\item [FRF 7.]     If an active medium or large excess node $i$ of $F$ is selected for pushing at the $\Delta$-scaling phase, then $\frac{\Delta}{k}$ units of flow are sent from $i$ to a violating node $v \in Desc(i, F)$.   (Specified in the procedure \textit{FRF-Push}.) 
\item [FRF 8.]   The number of scaling phases in which the flow-return forest is non-empty is $O(m \log_k n)$.  (See Corollary \ref{cor:FRFphases}.)
\end{enumerate}

\bigskip

\subsection{Operations on the FRF}
\label{subsec:operations}

In this subsection, we describe the operations that are performed on the FRF.  We also indicate how the operations are implemented and provide their total running time.  We describe all of the procedures in terms of their ``Applicability'' and ``Action.''     The procedures {\it FRF-Add} and {\it FRF-Delete} have additional details that are provided in their pseudo-codes.\\
 
\noindent  \textbf{\textit{Procedure FRF-Initialize}}$(F)$   \\
Applicability:   Initialization of the algorithm.\\
Action:   This procedure sets $F$ to be the null forest.\\

\bigskip

\noindent  \textbf{\textit{Procedure FRF-Push}}$(v, F, \Delta)$\\
Applicability:  Node $v \in F$ has medium or large excess and $v$ is not a leaf node of $F$.\\
Action:   The procedure finds a violating leaf node $i \in Desc(v, F)$.  Then $\frac{\Delta}{k}$ units of flow are sent in the path from $v$ to $i$ in $F$.\\

\noindent  \textbf{\textit{Procedure FRF-Pull}}$(v, F, \Delta)$\\
Applicability: Node $v$ is violating at the beginning of the 
$\Delta$-scaling phase and one of the following two conditions is satisfied. (1) Node $v$ is still violating and  
$FlowNeeded(v) =\frac{\Delta}{k}$, or (2) $v$ is a leaf node of $F$ that is no longer violating and about to be deleted,  
and $\hat e(Root(v, F), \Delta)   \ge \frac{\Delta}{k}$. \\
Action: $ \frac{\Delta}{k}$ units of flow are sent on the path in $F$ from $Root(v, F)$ to node $v$.  \\

\noindent  \textbf{\textit{Procedure  FRF-Add}}$(v, F, \Delta)$ \\
Applicability:   Node $v$ is newly violating at the $\Delta$-scaling phase.   \\
Action:   This procedure calls Procedure {\it FRF-Reverse-DFS}$(v, F, \Delta)$, which 
either produces an FRF-eligible directed cycle $W$ or else it 
produces an FRF-eligible path $W$.  In the former case, the cycle 
$W$ is contracted and {\it FRF-Add}$(v, F, \Delta)$ is executed on the 
contracted graph.  In the latter case, $F$ is replaced by $F \cup 
W$, and the arrays $NeededFlow(\,)$ and $Reserve(\,)$ are updated.   \\

\noindent  \textbf{\textit{W := FRF-Reverse-DFS}}$(v, F, \Delta)$ \\
Applicability:  {\it FRF-Add}$(v, F, \Delta)$ has just been called.\\
Action:  Let $Y = F \cup  \{t\} \cup \{i\notin F : i \text{ is normal.} \} $.  The algorithm finds an FRF-eligible 
cycle $W$ such that $W \cap F = \emptyset$  or it finds an 
FRF-eligible path $W$ from a node $w \in Y$ to node $v$ such that $(W\setminus w) \cap F = \emptyset$. \\

\noindent  \textbf{\textit{Procedure FRF-Delete}}$(v, F, \Delta)$\\
Applicability:  Node $v$ is a leaf node of $F$ and $\hat e(v, \Delta) 
\ge 0$.\\
Action:   $FlowNeeded(v)$ and $Reserve(Root(v, F))$ are 
both updated.   If $\hat e(Root(v, F), \Delta) \ge \frac{\Delta}{k}$ and if 
$v$ was a violating node at the beginning of the $\Delta$-scaling phase, then $\frac{\Delta}{k}$ units of 
flow are sent from $Root(v, F)$ to $v$.  Node $v$ is then deleted from $F$.\\ 



We now provide additional details for the procedures \textit{FRF-Add}, \textit{FRF-Reverse-DFS}, and \textit{FRF-delete} .\\

\noindent  \textbf{\textit{Procedure FRF-Add}}$(v, F, \Delta)$\\
01.  {\bf begin}\\
02.  \hspace{.15 in} $W :=$ {\it FRF-Reverse-DFS}$(v, F, \Delta)$;\\
03.  \hspace{.15 in}  {\bf if} $W$ is a cycle,  {\bf then do} \\
04.  \hspace{.3 in}	contract the cycle $W$;\\
05.  \hspace{.3 in}	run {\it FRF-Add}$(v, F, \Delta)$ on the contracted network;\\
06.  \hspace{.15 in}  {\bf else continue}\\        
07.  \hspace{.3 in}     let $j$ be the first node of $W$;\\
08.  \hspace{.3 in}     {\bf if} $j \notin F$, {\bf then }$w := j$ and $Reserve(w):= 0$; {\bf else} $w := Root(j, F)$;\\
09.  \hspace{.3 in}     $F := F \cup W$; \\    
10.  \hspace{.3 in}    \textbf{for every} node $i \in W$ \textbf{do} \\
11.  \hspace{.45 in}  $Root(i, F) := w$;\\
12.  \hspace{.45 in}  \textbf{if }  $i$ is newly violating, \textbf{then do}\\
13.  \hspace{.6 in}   $NeededFlow(i) := \epsilon^2 \Delta$;\\
14.  \hspace{.6 in}  $Reserve(w) :=$ $Reserve(w) + \epsilon^2 \Delta$;\\
15.   {\bf end}	\\
 
In order to run \textit{FRF-Reverse-DFS}, one needs to be able to identify FRF-eligible arcs for the depth first search.  For every node $i$, we maintain an array called \textit{FRF-eligible}$(i)$, which consists of all FRF-eligible arcs directed 
into node $i$.  
When procedure {\it FRF-Reverse-DFS} is run, it scans at most the first arc of \textit{FRF-eligible}$(i)$ for each $i$.   Each call of {\it FRF-Reverse-DFS} takes $O(n)$  time.

The algorithm also maintains an array called $Abundant(i)$, which consists of the Abundant arcs directed into node $i$.  The array $Abundant(\,)$ is needed so that \textit{FRF-eligible}$(\,)$ can be efficiently maintained. 
$Abundant(\,)$ is updated whenever an arc becomes large, and whenever an anti-abundant arc becomes bi-abundant.  (See Section \ref{sec:biabundant} of the appendix for more details.) 
\textit{FRF-eligible}$(\,)$ is updated whenever $Abundant(\,)$ is updated and whenever a node is relabeled.   The total time to update  $Abundant(\,)$ and \textit{FRF-eligible}$(\,)$ is proportional to the number of pushes plus the number of scaling phases.\\
 
We now consider the procedure \textit{FRF-delete}.  The unusual feature of this procedure is the sending of additional flow to nodes prior to their being deleted from $F$.
 Let $V(\Delta)$ denotes the subset of nodes $v$ of $N(\Delta)$ such that $0 < NeededFlow(v) < \Delta/k$.  That is, node $v$ is  violating at the beginning of the $\Delta$-scaling phase, and $v$ is not scheduled for a ``pull'' during the $\Delta$-scaling phase.\\

\smallskip
 
\noindent    \textbf{\textit{Procedure  FRF-Delete}}$(v, F, \Delta)$\\
01.   {\bf begin} \\
02.  \hspace{.15 in} $w := Root(v,F)$; \\
03.  \hspace{.15 in} $\delta := \text{NeededFlow}(v)$;\\
04.  \hspace{.15 in} NeededFlow$(v) := 0$;\\
05.  \hspace{.15 in} $Reserve(w) :=$ $Reserve(w) - \delta$;\\
06.  \hspace{.15 in} {\bf if} $w \in V(\Delta)$ and $\hat e(w, \Delta) \ge \frac{\Delta}{k}$, 
 	{\bf then} call Procedure {\it FRF-Pull}$(v, F, \Delta)$;\\
07.  \hspace{.15 in}  delete node $v$ from $F$;\\
08.  {\bf end}\\

Suppose that a leaf node $v$ is deleted during the $\Delta$-scaling 
phase.    Let $i = Root(v, F)$, and suppose that $\hat e(i) \ge \frac{\Delta}{k}$ immediately prior to the deletion of node $v$ from $F$.  We claim that under these circumstances,  the deletion of $v$ from $F$ does not increase the modified excess of node $i$.    We consider three cases corresponding to different values of $NeededFlow(v)$ at the beginning of the $\Delta$-scaling phase: (1) $NeededFlow(v) = 0$, (2) $NeededFlow(v) =  \frac{\Delta}{k}$, and (3) $ 0 < NeededFlow(v) <  \frac{\Delta}{k}$.  \\

In Case (1), $Reserve(i)$ is not altered by the deletion of $v$, and thus the modified excess of node $i$ is not changed.   In Case (2), \textit{FRF-Pull}$(v, F)$ is called at the beginning of the $\Delta$-scaling phase.    This leads to a decrease in the modified excess of node $i$ by $\frac{\Delta}{k}$.   When node $v$ is later deleted, then $Reserve(i)$ is decreased by 
$\frac{\Delta}{k}$ in Step 05, thus increasing the modified excess of node $i$ by $\frac{\Delta}{k}$.  The increase and decrease cancel each other.  In Case (3), 
the modified excess of $i$ will increase by less than $\frac{\Delta}{k}$ in Step 05.  In Step 06,  $\frac{\Delta}{k}$ units of flow will be sent from $i$ to $v$, which reduces the modified excess of $i$  by $\frac{\Delta}{k}$.  The decrease exceeds the increase.   Thus, in all three cases, the deletion of node $v$ does not lead to an increase in the modified excess of node $i$.      

\begin{thm}   
\label{th:FRFflows}
Consider the procedures \textit{FRF-ADD}, \textit{FRF-Reverse-DFS}, \textit{FRF-Delete}, \textit{FRF-Push}, and \textit{FRF-Pull}.   The bounds given in Table \ref{FRFtable} are correct.   The total time taken by all of these procedures is $O(nm)$.
\end{thm}  

\begin{table}[h]
  \includegraphics[width=5 in]{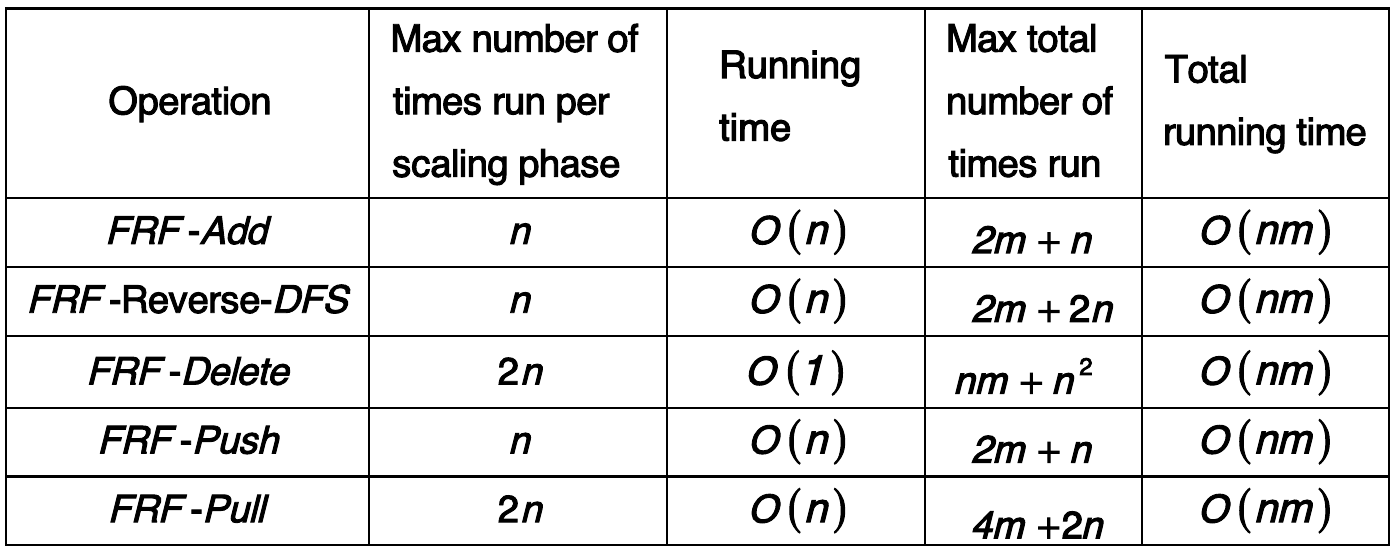}
  \caption{ FRF-Procedures.}
  \label{FRFtable}
\end{table}

\begin{proof}

In the proof of Lemma \ref{lem:newlyviolating}, we showed that there are at most $2m + n$ newly violating nodes added to $F$ over all scaling phases.   

 \textit{FRF-ADD} is run when a node becomes newly violating and when \textit{FRF-Reverse-DFS} creates a contraction.    There are at most $n$ newly violating nodes in a scaling phase, and at most $2m + n$ violating nodes over all scaling phases.   There are fewer than $n$ contractions.  The time for each execution of  \textit{FRF-ADD} is $O(n)$. 

\textit{FRF-Reverse-DFS} is run whenever  \textit{FRF-ADD} is run.    Each time that it is run takes $O(n)$ time. 

When \textit{FRF-Delete} is called, the running time is $O(1)$, ignoring the time for \textit{FRF-Pull}, which is accounted for elsewhere.  The number of times that \textit{FRF-Delete}  is called is the total number of nodes that are added to the flow-return forest, which is at most $n(m + n)$.

Whenever \textit{FRF-Push} is called, a violating node becomes non-violating.  This can happen at most $2m + n$ times in total.  The running time of procedure \textit{FRF-Push} is $O(n)$.

For each violating node $v$, \textit{FRF-Pull} can be called at most twice.  It can be called once if Needed-Flow$(v) = \frac{\Delta}{k}$.   It can be called once during \textit{FRF-Delete$(v, F)$} if $0 < \text{Needed-Flow}(v) < \frac{\Delta}{k}$.   Thus, \textit{FRF-Pull} can be called at most $4m + 2n$ times.  The running time of procedure \textit{FRF-Pull} is $O(n)$.
\end{proof}

%
%
%
%

\subsection{The number of scaling phases in which $F$ is non-empty}

In this subsection, we prove that there are $O(m \log_k n)$ scaling phases in which $F$ is non-empty.  We also show that a violating node becomes non-violating prior to having excess as low as $-\epsilon \Delta$.

\begin{lem}
\label{lem:2Q}
Each node is violating for at most $2Q+1$ phases of the Enhanced LMES Algorithm.
\end{lem}
\begin{proof}
Suppose that node $v$ is newly violating in the $\Delta'$-scaling phase.  Then, 
$FlowNeeded(v) = \frac{\epsilon^2\Delta'}{k}$.  We will show that $v$ becomes non-violating within $2Q$ additional scaling phases.

Now consider Line 03 of the procedure \textit{ScalingPhase}.  If $v$ is still violating at the $\Delta$-scaling phase and if  $FlowNeeded(v) = \frac{\Delta}{k}$, then $ \Delta = \epsilon^2\Delta'$.  There has been exactly $2Q$ phases since node $v$ was added to $F$.  Then the algorithm runs \textit{FRF-Pull}$(v, F, \Delta)$, after which $v$ is no longer violating.
\end{proof}

\begin{cor}
\label{cor:FRFphases}
   The number of scaling phases in which $F \neq \emptyset$  is 
$O(m \log_k n)$.
   \end{cor}

\begin{proof}   
   If $F \neq \emptyset$, then $F$ has a violating node.   By Lemma 
 \ref{lem:newlyviolating}, there are $O(m)$ newly violating nodes over all scaling 
   phases.  By Lemma \ref{lem:2Q}, each violating node is in $F$ for at most $2Q+1$ scaling 
   phases.  Thus the number of scaling phases in which $F \neq 
   \emptyset$ is $O(m \log_k n)$.  
   \end{proof}

\begin{cor}
\label{cor:2Q}
Suppose that node $v$ is violating at the beginning of the $\Delta$-scaling phase.  Then 
 $-\epsilon \Delta < e(v, \Delta) < \epsilon \Delta$.
\end{cor}
\begin{proof}
For node $v$ to be violating, $ e(v, \Delta) < \epsilon \Delta$.   We need to prove that $ e(v, \Delta) > -\epsilon \Delta$. 
Suppose that $v$ became newly violating at the $\Delta'$-scaling phase when the preflow was $x$.  Suppose $v$ becomes non-violating at the $\Delta$-scaling phase when the preflow is $y$.  Then $e_y(v) \ge e_x(v)$ since there is no push from $v$ in the iterations leading up to preflow $y$.   

By Lemma \ref{lem:2Q}, $\Delta \ge \epsilon^2 \Delta'$.  
At the end of the $(k\Delta')$-scaling phase (the phase prior to the $\Delta'$-scaling phase), node $v$ was special. Thus
$ e_y(v) \ge e_x(v) \ge -1.5 \epsilon^4 (k \Delta') > -\epsilon ^3 \Delta'  \ge -\epsilon  \Delta  $.   
\end{proof}

\smallskip

\section{The Enhanced LMES max flow algorithm}
\label{sec:enhanced}

In this section, we give the pseudo-code for the remaining procedures of the Enhanced LMES Algorithm.  We also prove the correctness of the Enhanced LMES Algorithm and prove that the running time is 
$O(k n^2 + nm \log_k n)$.

\subsection{The Pseudo-code}
\label{subsec:pseudo}

We have already provided the pseudo-code for the Enhanced-LMES Algorithm including the main procedure, which is called \textit{ScalingPhase}.  Here, we provide the pseudo-code the remaining procedures: \textit{Enhanced-LMES-Push/Relabel}, \textit{FRF-Recursive-Delete-and-Merge}, and \textit{Get-Next-Scaling-Parameter}.\\

\noindent \textbf{\textit{Enhanced-LMES-Push/Relabel}$(v, \Delta)$}\\
01.   \textbf{begin}\\
02.   	\hspace{.15 in} \textbf{if} there is no admissible arc in $A^+(v)$, \textbf{then } $d(v) := d(v) + 1$;	\\
03.  \hspace{.15 in}  \textbf{else } select an admissible arc $(i, j) \in A^+(v)$;\\
04.     \hspace{.15 in} $ \textit{Enhanced-LMES-Push}(i, j, \Delta)$ \\
05.   \textbf{end}\\

We next describe Procedure \textit{FRF-Recursive-Delete-and-Merge}. An arc $(i, j)$ with distinct endpoints is called \textit{mergeable} if $(i, j)$ is doubly abundant and at most one endpoint of $(i, j)$ is in $F$.  The Enhanced LMES Algorithm  contracts mergeable bi-abundant arcs.  It does not contract a bi-abundant arc if both endpoints of the arc are in $F$.   (This restriction ensures that $F$ remains a forest.)  If both endpoints are in $F$, then the algorithm defers contraction until one or both of the endpoints are deleted from $F$.  

When a node $v$ is deleted from $F$, the algorithm checks whether $v$ is incident to a bi-abundant arc.  If so,  the arc is contracted.   When $v$ is deleted from $F$, the algorithm also checks the parent of $v$ in $F$.  If the parent node has become a leaf node that is non-violating, then the procedure deletes the parent node.   \\

\noindent \textbf{\textit{Procedure FRF-Recursive-Delete-and-Merge$(F, \Delta)$}}\\
01.   \textbf{begin}\\
02.   	\hspace{.15 in}\textbf{while} there is a mergeable arc or a non-violating leaf node of $F$ 	\textbf{do}\\
03.   	\hspace{.30 in}	\textbf{if} there is a non-violating leaf node $v$ of $F$ 
\textbf{then} \textit{FRF-Delete}$(v, F, \Delta)$\\
04.     \hspace{.30 in} 		\textbf{else}  contract a mergeable arc;  \\
05.	\hspace{.15 in}  \textbf{endwhile}\\
06.   \textbf{end}\\


We now describe the procedure that chooses the subsequent scaling parameter.  In general, the scaling parameter after the $\Delta$-scaling phase is $\frac{\Delta}{k}$.   However, this could lead to a long number of scaling phases that are useless; that is, scaling phases with no push and such that $F= \emptyset$. The procedure \textit{Get-Next-Scaling-Parameter} avoids long sequences of useless phases, while also maintaining  Invariant 2.

Suppose that the current scaling parameter is $\Delta$.  Suppose further that the next $2Q+1$ scaling phases are going to be useless.  Then the scaling parameter is set to a value $\Gamma$ so that the $\Gamma$-scaling phase is useful and so that $\Gamma \le \frac{\Delta}{Mk}$.  \\

\noindent \textbf{\textit{Procedure Get-Next-Scaling-Parameter$(\Delta)$}}\\
01.   \textbf{begin}\\
02.   	\hspace{.15 in} $\Gamma_1 := \max\{e(v) : e(v) > 0\}$;\\
03.   \hspace{.15 in} 	$\Gamma_2 := \max\{-\epsilon^3 e(v) : e(v) < 0\}$;\\
04.    \hspace{.15 in}   	$\Gamma := \max\{\Gamma_1, \Gamma_2\}$;  \\
05.	 \hspace{.15 in}    \textbf{if} $\Gamma < \frac{\Delta}{kM}$ , \textbf{then} $ \Delta := \Gamma$;   \textbf{else}, 
$\Delta := \frac{\Delta}{k}$;\\
06.   \textbf{end}\\

Suppose that $\Delta$ is replaced by $\Gamma$ in Step 05.    We claim that Invariant 2 is satisfied at the beginning of the $\Gamma$-scaling phase .  Suppose that $r_{ij} > 0$ and $r_{ji} > 0$ at the end of the $\Delta$-scaling phase.  Then it is also true that $r_{ij} \ge \frac{\Delta}{k}$, and $r_{ji} \ge \frac{\Delta}{k}$.  Then $r_{ij} \ge M \Gamma$ and $r_{ji} \ge M \Gamma$, and $(i, j)$ is  bi-abundant at the $\Gamma$-scaling phase.  This choice of $\Gamma$ ensures that Invariant 2 remains satisfied.\\


%

\section{ The Enhanced LMES Algorithm runs in $O(k n^2 + nm \log_k n)$ time}
\label{sec:analysis}

\subsection{The main theorem and its supporting lemmas}

In this subsection, we state the main theorem as well as the lemmas that are used in establishing the theorem.
In subsequent subsections, we prove the lemmas and the main theorem.\\

\noindent \textbf {Theorem \ref{th:main}.} 
\textit{The Enhanced LMES Algorithm finds the maximum flow in $O(k n^2 + nm \log_k n)$ time.  Let $k$ be the smallest power of 2 that is at least $\max\{\log\log n, \frac{m}{n}, 4\}$.   Then the running time is 
$$O\left(\frac{nm \log n}{\log\log n + \log \frac{m}{n}}\right).$$} 

\smallskip

\noindent  \textbf{Lemma \ref{lem:normal}.}
\textit{The total number of occurrences of normal and violating nodes in the Enhanced LMES Algorithm is 
$O(m \log_k n)$.}\\

\smallskip

\noindent \textbf {Corollary \ref{cor:phases}.}  
\textit{The number of scaling phases of the Enhanced LMES Algorithm is $O(m \log_k n)$.}\\

\smallskip

\noindent  \textbf{Corollary \ref{cor:terminate}}.
\textit{The Enhanced LMES Algorithm terminates with a maximum flow.}\\

\smallskip

\noindent  \textbf{Lemma \ref{lem:ehat}.}
\textit{If $v \in N(\Delta)$, then throughout the   $\Delta$-scaling phase  of the Enhanced LMES Algorithm, 
$\hat e(v) < (1 + (k-1) \epsilon ) \Delta$  }. \\

\smallskip

\noindent  \textbf{Lemma \ref{lem:totalflow}.}
\textit{The total amount of flow pushed in the $\Delta$-scaling phase of the Enhanced LMES Algorithm is at most $5n^2 \Delta$.}\\

\smallskip

\noindent \textit{\textbf{Corollary \ref{cor:abundance}}.
If $(i, j)$ is abundant at the beginning of the $\Delta$-scaling phase of the Enhanced LMES, then for every scaling parameter $\Delta' \le \Delta $ the following are true. 
\begin{enumerate}
\item   $(i, j)$ is abundant in the $\Delta'$-scaling phase and
\item   $r_{ij} > 0$ throughout the $\Delta'$-scaling phase.
\end{enumerate}}

\smallskip

\noindent  \textbf{Lemma \ref{lem:totalpushes}.}
\textit{ The total number of medium and large pushes in the Enhanced LMES Algorithm is $O(k n^2 + nm \log_k n)$.}\\
 
\smallskip

\subsection{The number of normal and violating nodes}

\begin{lem}
\label{lem:normal}
The total number of occurrences of normal and violating nodes in the Enhanced LMES Algorithm is 
$O(m \log_k n)$.
\end{lem}

\begin{proof}  
 If a node $v$ is normal, then (by the Contraction Lemma),  node $v$ will become incident to a medium arc or a 
 bi-abundant arc within $7Q$ scaling phases.  Thus the number of occurrences of normal node that are not incident to a bi-abundant arc 
 is $O( m \log_k n)$.
 
We now claim that there are $O( m \log_k n)$  occurrences of normal 
nodes that are incident to a bi-abundant arc.   Each of these nodes must be in $F$; otherwise the bi-abundant arc would be contracted. 
If $v \in F$ is normal at the $\Delta$-scaling phase, then 
$e(v) \ge \epsilon \Delta$, and node $v$ will have medium or large excess 
within $Q$ scaling phases, after which there will be an FRF-push from node $v$.   By Theorem \ref{th:FRFflows},   the total number of FRF-pushes is $O(m)$.   Therefore, there are   $O( m \log_k n)$ occurrences of normal nodes that are incident to bi-abundant arcs.  

Finally, we consider occurrences of violating nodes, all of which are added to $F$ when they become newly violating.  By Lemma \ref{lem:newlyviolating}, at most $2m+n$ newly violating nodes were added to $F$. After being added to $F$, a violating node will receive flow within $2Q$ scaling phases, at which point it will become normal. Thus, there are $O(m \log_k n)$ occurrences of violating nodes.
\end{proof}

\subsection{The number of scaling phases}

We say that a scaling phase is \textit{useful} if there is at least one push during the scaling phase.

\begin{lem}
\label{lem:usefulphases}
The number of useful scaling phases of the Enhanced LMES Algorithm is $O(m \log_k n).$
\end{lem}
	
\begin{proof}
By Corollary \ref{cor:FRFphases}, the number of scaling phases in which $F \neq \emptyset$  is $O(m \log_k n)$.  By Lemma \ref{lem:normal}, the number of phases in which there is a normal node is $O(m \log_k n)$.  If there is no normal node, and if $F = \emptyset$, then there is no push during the scaling phase, and the scaling phase is useless.
\end{proof}

\begin{cor}
\label{cor:phases}
The number of  scaling phases of the Enhanced LMES Algorithm is $O(m \log_k n)$.
\end{cor}
\begin{proof}
We need to consider scaling phases in which $F = \emptyset$ and no node is normal.  
This implies that every node must be special except for $s$ and $t$.   We  claim that there will be a new normal node or a new violating node within $2Q+2$ scaling phases.
Let  
$\Gamma  = \max \{\Gamma_1 , \Gamma_2 \}$, as selected in Step 04 of \textit{Get-Next-Scaling-Parameter}.  Consider first the case that $\Gamma < \frac{\Delta}{kM}$.   In this case, the next scaling parameter is $\Gamma$.
If $\Gamma  = \Gamma_1$, then there will be a normal node at the $\Gamma$-scaling phase.  If $\Gamma  = \Gamma_2$, then there will be violating node at the $\Gamma$-scaling phase, after which $F$ becomes non-empty.   Therefore, the total number of scaling phases is  $O(m \log_k n)$.

Now consider the case that $\Gamma > \frac{\Delta}{kM}$.   In this case, we continue to divide scaling parameters by $k$ until there is a normal node or a violating node.  This will occur within $2Q+2$ scaling phases, and which point the scaling parameter is $\Delta'$ and $\Delta' < \frac{\Delta}{kM}$.

The number of newly normal nodes is bounded at most $n$ plus the number of newly special nodes, which is  $O(m)$.   The number of newly violating nodes is at most the number of newly special nodes, which is $O(m)$.  Thus, the number of useless phases is 
$O(m \log_k n)$.
\end{proof}

\begin{cor}
\label{cor:terminate}
The Enhanced LMES Algorithm terminates with a maximum flow.
\end{cor}

\begin{proof} 
  The algorithm must terminate since there are a bounded number of scaling phases.   The termination criterion is that there are no nodes with excess, which is the optimality criterion for the generic push-relabel algorithm.   The flow is optimum in the contracted graph.   Subsequently, the contracted abundant cycles are expanded, resulting in an optimal flow in the original graph.  
  \end{proof}

\subsection{An upper bound on $\hat e(v, \Delta)$}

Let $N(\Delta)$ be the set of nodes (including merged nodes) at the beginning of the $\Delta$-scaling phase. In this section, we show that for all nodes $v \in N(\Delta)$,   $\hat e(v, \Delta) \le \Delta + (k-1) \epsilon \Delta $.  This upper bound would be obvious except for the impact of contractions on node excesses.

Suppose now that $w$ is a merged node that is created  during the $\Delta$-scaling phase by the contraction of an abundant cycle $W$.   When $w$ is created,  
$e(w) = \sum_{v\in W} e(v)$. 

It is possible that $e(w)$ is proportional to $|W| \Delta$, which violates the upper bound on excess that is  needed for the complexity argument.   The assumption that $e(v) = O(\Delta)$ is needed to bound the increase in potential functions for relabels of nodes.   We bypass this difficulty by evaluating potential functions in the $\Delta$-scaling phase by summing over nodes of $N(\Delta)$ and not including the merged nodes obtained by contractions during the $\Delta$-scaling phase.  

Suppose that $w$ is created by the merging of one or more abundant cycles during the $\Delta$-scaling phase.  Let $W$ denote the nodes of $N(\Delta)$ that were nodes of any of the cycles that were merged into node $w$. 

For a given preflow $x$, for all $v \in W$, we define 
$\hat e_x(v, \Delta)$ as follows:
\begin{enumerate}
\item If $v \notin Roots$, then 
$\hat e_x(v) = \frac{e_x(w) - \epsilon \Delta }{|W|}$.
\item If $v \in Roots$, then 
$\hat e_x(v) = \frac{e_x(w) - \epsilon \Delta - Reserve(v)}{|W|} $.
\end{enumerate}

\begin{lem}
\label{lem:ehat}
If $v \in N(\Delta)$, then throughout the   $\Delta$-scaling phase  of the Enhanced LMES Algorithm, $\hat e(v, \Delta) < (1 + (k-1) \epsilon ) \Delta$.  
\end{lem} 

\begin{proof}
The operations that can lead to an increase in $\hat e(v)$ are the following:   (1) a  push into node $v$, (2) the beginning of a scaling phase, and (3) a contraction including node $v$ or a node that $v$ has been merged into.  

\begin{itemize}
\item [(1)]  If there is a push into node $v$, then $\hat e(v, \Delta) < \frac{\Delta}{2}$ prior to the push, and $\hat e(v, \Delta) < \Delta$ after the push.  

\item  [(2)] At the end of the previous scaling phase (i.e, the $(k\Delta)$-scaling phase),  
$\hat e(v, k\Delta) = e(v) - \epsilon k \Delta < \frac{(k\Delta)}{k} = \Delta$.  At the beginning of the $\Delta$-scaling phase, $\hat e(v,\Delta) = e(v) - \epsilon \Delta < \Delta + (k-1) \epsilon \Delta $.

\item [(3)]  Suppose that node $w$ is a contracted node that is formed in the $\Delta$-scaling phase.  Let $W$ denote the set of nodes of $N(\Delta)$ that are contracted into node $w$.   For each node $v \in W$, let  $\bar e(v, \Delta)$ denote the modified excess prior to contraction, and let $\hat e(v, \Delta)$ denote the modified excess after contraction.   Assume inductively that for each node $v \in W$, 
 $\bar e(v, \Delta) < (1 + (k-1) \epsilon ) \Delta$.   Then  
 
$$  \hat e(v, \Delta) < \sum_{v \in W} \frac{\hat e(v, \Delta)}{|W|} 
< (1 + (k-1) \epsilon ) \Delta.  $$
\end{itemize}
This completes the proof.
 \end{proof}

\subsection{Bounding the total flow sent in a scaling phase}
\label{totalflowsection}

In this subsection, we  bound the total flow in a scaling phase.  In order to bound the total flow sent in a scaling phase, we  replace the potential function $\Phi_1$ used in Section \ref{sec:LMES} with a similar potential function  $\Phi_3(\Delta)$ defined next.   Recall that $N(\Delta)$ is the set of nodes (both original and merged) at the beginning of the 
$\Delta$-scaling phase.   Let $R(\Delta)$ denote the set of root nodes of $F$ at the beginning of the $\Delta$-scaling phase.

$$  \Phi_3(\Delta) = \sum_{v \in N(\Delta) \backslash R(\Delta) }  \frac{e(v)d(v)}{\Delta} 
+ \sum_{v \in  R(\Delta) }  \frac{(e(v) - Reserve(v))\cdot d(v)}{\Delta}.$$

\begin{lem}
\label{lem:totalflow}
The total amount of flow pushed in any arc in the $\Delta$-scaling phase of the Enhanced LMES Algorithm is at most $5n^2 \Delta$.
\end{lem} 

\begin{proof}
The flows in  the $\Delta$-scaling phase can be  (1) pushes in admissible arcs or (2) flow sent on paths of $F$ in the procedures \textit{FRF-Push} and \textit{FRF-Pull}.   We refer to the latter as FRF-flows.   

We first consider the FRF-flows. In the proof of Theorem \ref{th:FRFflows}, \textit{FRF-Delete}  can be called at most $2n$ times in a scaling phase.  Accordingly, \textit{FRF-Push}  can be called at most $2n$ times during a scaling phase.  And \textit{FRF-Pull}  can be called at most $4n$ times during a scaling phase.
Thus the total amount of flow sent in  \textit{FRF-Push}  and \textit{FRF-Pull} per scaling phase is less than $ \frac{6n\Delta}{k}$, which is less than $1.5 n \Delta$.

We now consider flow sent in admissible arcs. 
Each push of $\delta$ units in an admissible arc reduces $\Phi_3(\Delta)$ by $\frac{\delta}{\Delta}$.   Therefore, the total flow in the $\Delta$-scaling phase divided by $\Delta$ is at most the total decrease in
$\Phi_3(\Delta)$, which is at most the initial value of $ \Phi_3(\Delta) $  plus the total increase in $\Phi_3(\Delta)$ during the scaling phase.  

The initial value of $\Phi_3(\Delta)$ is less than $(n+1)(n-2)(1 + (k-1)\epsilon) < n^2$.
The increase in $\Phi_3(\Delta)$ due to relabels is less than $(n+1)(n-2)(1 + (k-1)\epsilon ) < n^2$.

Contractions do not lead to an increase in $\Phi_3(\Delta)$.  However, calls of \textit{FRF-Push}  and \textit{FRF-Pull} may lead to an increase in $\Phi_3(\Delta)$.  In addition, decreases in $Reserve(v)$ for $v \in R(\Delta)$  lead to an increase in  $e(v) - Reserve(v)$, and thus to an increase in $\Phi_3(\Delta)$.   

When sending flow from node $i$ to node $v$ in $F$, it is possible that $d(i) < d(v)$, and that $\Phi_3(\Delta)$ will increase.  When  \textit{FRF-Push}  or \textit{FRF-Pull}  is called,  that maximum possible increase in $\Phi_3(\Delta)$ is at most $\frac{n}{k}$.  The total increase in $\Phi_3(\Delta)$ due to FRF flows is at most   $6n \cdot \frac{n}{k}$, which is at most $1.5 n^2$.  

$Reserve(\,)$ is decreased only when a previously violating  node becomes non-violating and is deleted from $F$.  Each decrease in $Reserve(\,)$ leads to an increase in $\Phi_3(\Delta)$ by at most $\frac{n}{k}$.   This occurs at most $2n$ times in a scaling phase.  The total increase in $\Phi_3(\Delta)$ due to decreases in $Reserve(\,)$ is at most  $2n \cdot \frac{n}{k}$, which is at most $.5 n^2$. 

Thus, the total increase in $\Phi_3(\Delta)$ during the $\Delta$-scaling phase is less than $3n^2$. The initial value of  $\Phi_3(\Delta)$ plus the total increase in $\Phi_3(\Delta)$ during the $\Delta$-scaling phase is less than 
$4n^2$. And the total amount of flow sent in the $\Delta$-scaling phase is pushes in any arc $(i, j)$ is most $4 n^2 \Delta$.   If we include the contribution of FRF-pushes as well, then the total amount of flow sent in any arc $(i, j)$ is most $5 n^2 \Delta$. 
\end{proof}

\begin{cor}
\label{cor:abundance}
If $(i, j)$ is abundant at the beginning of the $\Delta$-scaling phase of the Enhanced LMES, then for every scaling parameter $\Delta' \le \Delta $ the following are true. 
\begin{enumerate}
\item   $(i, j)$ is abundant in the $\Delta'$-scaling phase and
\item   $r_{ij} > 0$ throughout the $\Delta'$-scaling phase.
\end{enumerate}
\end{cor}

\begin{proof}
If $(i, j)$ is abundant at the beginning of the $\Delta$-scaling phase, then $r_{ij} \ge M \Delta$ at the beginning of the phase.  By Lemma \ref{lem:totalflow}, 
$r_{ij} \ge (M - 5n^2) \Delta > .5 M \Delta $ at the end of the $\Delta$-scaling phase.  It follows that $(i, j)$ is abundant at the beginning of the next scaling phase. 
\end{proof}

\bigskip

\subsection{The total number of pushes}

In this subsection, we bound the total number of pushes.

\begin{lem}
\label{lem:totalpushes}
 The total number of medium and large pushes in the Enhanced LMES Algorithm is $O(k n^2 + nm \log_k n)$.
\end{lem} 

\begin{proof}

Let $\Psi $ be the set of distinct scaling parameters.   By Lemma \ref{cor:phases}, $|\Psi| = O(m \log_k n)$.

We define $\Phi_4$ as a modification of  $\Phi_2$ in Section \ref{sec:LMES}.  The parameters $P$ and $\ell$ are defined in the same manner as for $\Phi_2$ except that $P$ is restricted to nodes of $N(\Delta)$ and does not include any nodes merged at the $\Delta$-scaling phase.

$$\Phi_4(\Delta):=\sum_{j\in P\backslash R(\Delta)}  e(j)\cdot\frac{d(j)-\ell+1}{\Delta} + 
\sum_{j\in P\cap R(\Delta)}  (e(j) - Reserve(j))\cdot\frac{d(j)-\ell+1}{\Delta}.$$ 
\smallskip

For each $\Delta \in \Psi$, let $S(\Delta)$ denote the subset of nodes of $N(\Delta)$ that are special at the beginning of the $\Delta$-scaling phase.   The  observations in Lemma \ref{lem:facts} are straightforward and will be used in the proof of Lemma \ref{lem:totalpushes}  

\begin{lem}
\label{lem:facts}
The following statements are all true.
\begin{enumerate}  
\item      [Fact 1. ] If $v \in S(\Delta)$, then $e(v) \le \epsilon^4 \Delta $ at the beginning of the $\Delta$-scaling phase.  
\item [Fact 2. ] If $v \in S(\Delta)$, then $v$ remains special throughout the $\Delta$-scaling phase.  
\item  [Fact 3. ] If $v \in S(\Delta)$ and if  $v \in P$ at an iteration at which $\ell$ is reduced, then $e(v) \le 1.5 \epsilon^4 \Delta $  at that iteration. 
\item  [Fact 4. ] $\sum_{\Delta\in\Psi}  | N(\Delta)\backslash S(\Delta)| = O(m \log_k n)$.
\end{enumerate}
\end{lem}

\bigskip

We now bound the number of large pushes using $\Phi_3$.   
The number of large pushes is at most 2 times the total increase in $\Phi_3$ over all scaling phases.   
The value of $\Phi_3$ at the beginning of the $\Delta$-scaling phase is at most 

$$ \sum_{v\in N(\Delta)\backslash S(\Delta)}  \frac{d(v)e(v)}{\Delta} 
< (n+1) \cdot | N(\Delta)\backslash S(\Delta)|.$$ 

\smallskip

Thus, the sum of all increases in $\Phi_3$ at the beginning of scaling phases is less than 
\smallskip
$$  \sum_{\Delta \in \Psi } (n+1)|N(\Delta)\backslash S(\Delta)| =   O(nm \log_k n).$$    

The sum of the increases in $ \Phi_3$ over all scaling phases due to relabels is $O(n^2)$.  

Finally, we consider increases in $\Phi_3$ due to pushes in $F$ and decreases in $Reserve(v)$ .
Since there are at most $O(m)$ violating nodes in total, there are $O(m)$ FRF-flows and decreases in $Reserve(\,)$.   Each of these operations can increase $\Phi_3$ by at most $\frac{n}{k}$.  The total increase in $\Phi_3$  due to FRF operations is $O(\frac{nm  }{k})$.  
 
Taking into account all increases in $\Phi_3$, we conclude that the number of large pushes is  $O(n^2 + nm \log_k n)$.\\

We now bound the number of medium pushes using $\Phi_4$.  Recall that the initial value of $\Phi_4(\Delta)$ is 0 in each scaling phase.

Every medium push decreases $\Phi_4$ by at least $\frac{1}{k}$ except for the first medium push from a node.   Since there are $O(m \log_k n)$ scaling phases and there are at most $n$ nodes per phase, there are $O(nm \log_k n)$ first medium pushes from nodes.  (It is perhaps surprising that the first medium pushes from nodes is a bottleneck operation.) 

The remaining number of medium pushes is at most $k$ times the total increase in $\Phi_4$ over all scaling phases.  The potential function $\Phi_4$ can increase in the following situations.   (1)  there is a relabel of a node, (2) $\ell$ decreases, (3) a node is added to $P$, (4) there is an operation of \textit{FRF-Push} or \textit{FRF-Pull}, or (5) there is an operation of \textit{FRF-delete} in which $Reserve(\,)$ is decreased.

The total contribution to the increase in $\Phi_4$ due to relabels is at most $n^2$.   The contribution due to FRF operations (i.e., (4) and (5)) is the same as the corresponding contribution to $\Phi_3$, which is $O(\frac{nm}{k})$.  We now consider decreases to $\ell$ and additions to $P$.  

By 
Lemma \ref{lem:facts}, special nodes cause an increase in $\Phi_4$ of less than 1 over all scaling phase.  

Finally, we analyze the contributions of non-special nodes when $\ell$ decreases and when a node is added to $P$.    
For each decrease of $\ell$ by 1, the contribution of each node $v \in N(\Delta)\backslash S(\Delta)$ is at most 
 $ \frac{1}{k}$ to the increase in $\Phi_4$.   Thus, the total contribution of $\Phi_4(\Delta)$  due to all decreases in $\ell$ during the phase is at most $\frac{n}{k} 
| N(\Delta) \backslash S(\Delta)| $.  The increase in $\Phi_4$  due to all decreases in $\ell$ over all phases is less than 
$$ \sum_{\Delta\in\Psi}  \frac{n }{k} \cdot | N(\Delta)\backslash S(\Delta)| = 
O\left(\frac{nm \log_k n}{k}\right).$$  

Finally, we consider the increases in $\Phi_4$ when normal nodes are added to $P$.   A node $v$ can be added to $P$ in one of three ways:  (i) node $v$ is relabeled, (ii)  the parameter $\ell$ is reduced, and (iii) there is a medium push from node $v$.   We have already accounted for increases due to (i) and (ii) in our analysis.  We now consider (iii).   If a normal node $v$ is added to $P$ following a medium push, then $e(v) \le \epsilon \Delta +  \frac{\Delta}{k}$ following the push.  This leads to an increase in $\Phi_4$ of $O(\frac{1}{k})$.  Summing this increase over all normal nodes in all phases, we obtain an increase of $O(\frac{m \log_k n}{k})$.  

%
%

We conclude that the total increases in $\Phi_4$ over all scaling phases is $O(n^2 + \frac{nm \log_k n}{k}$).  This shows that the number of medium pushes over all scaling phases is $O(kn^2 + nm \log_k n)$.  
 \end{proof}
 
 \subsection{The main theorem}

\begin{thm}
\label{th:main}  
The Enhanced LMES Algorithm finds the maximum flow in $O(k n^2 + nm \log_k n)$ time.  Let $k$ be the smallest power of 2 that is at least $\max\{\log\log n, \frac{m}{n}, 4\}$.   Then the running time is 
$$O\left(\frac{nm \log n}{\log\log n + \log \frac{m}{n}}\right).$$

\end{thm}  

\begin{proof}
The optimality of the algorithm was established in Corollary \ref{cor:terminate}.
By Lemma \ref{lem:totalpushes}, the number of medium and large pushes is 
$O(k n^2 + nm \log_k n)$.  The time for contractions is $O(nm)$.   The time for all operations in the flow-return forest is 
$O(nm)$.  The time to update all arrays and lists at the beginning of a scaling phase is $O(n)$ per scaling phase and $O(nm \log_k n)$ in total.

The time for all other updates of arrays and lists is $O(nm + \text{\# of pushes})$.   We conclude that the total running time is $O(nm \log_k n)$.  If we choose $k$ to be the least power of 2 that is at least $\max \{\log\log n, \frac{m}{n}, 4 \}$, we obtain the optimized running time.
\end{proof}


\bigskip

\section{Summary and conclusions}
\label{sec:LMESC}

In this paper, we have developed Enhanced LMES Algorithm, which is a strongly polynomial time algorithm for the max flow problem with a running time of 
$$O\left(\frac{nm \log n}{\log\log n + \log \frac{m}{n}}\right).$$

In the case that $\frac{m}{n} \ge n^{1/16}$, our algorithm dominates Orlin's \cite{O13} algorithm, which runs in  
$O(nm + m^{31/16} \log^2 n)$ time.   Our algorithm strictly dominates the best other strongly polynomial time algorithm for the max flow problem,   which is due to King et al \cite{KRT94}.  In the case that $m = O(n \log n)$, the running time of our algorithm is faster than the algorithm in \cite{KRT94} by a factor of $\log\log n$. 

The  LMES Algorithm is based on  the Stack Scaling Algorithm of of Ahuja et al. \cite{AOT89}.  The Enhanced LMES Algorithm exploits the Contraction Lemma of   \cite{O13} in order to improve the running time.

In order to use the Contraction Lemma to decrease the running time, we introduced several innovations in this paper.  We first made a technical change to the amount of flow pushed so that anti-abundant arcs had residual capacities that were multiples of $\frac{\Delta}{k}$.  More significantly, we permitted a push to result in a slightly negative excess under restricted conditions.  We then introduced a data structure called the flow-return forest, which ensured that the excess of a node would never become ``too negative.''
The flow-return forest is potentially of value for other max flow algorithms or possibly for minimum cost flow algorithms. 

We  explored whether we could obtain an improved running time using dynamic trees.   For problems in which $m$ is close to $n$, there might possibly bea small  improvement in running time.   However, we did not see how to implement dynamic trees so as to improve the running time.   It is an open question whether such a speed-up is possible.\\

\noindent \textbf {Acknowledgements.}   This research was partially supported by the Office of Naval Research grant
N00014-17-1-2194.   We thank Laci Vegh for providing useful suggestions on how to improve the manuscript.   \\

\begin{appendix}

\section{Expanding contracted cycles}
\label{sec:expansion}

After a max flow in the contracted network is obtained, the contracted cycles are expanded in the reverse order in which they were contracted. Subsequently, the flow in the expanded network is a maximum flow.   The time for expansion of these cycles is $O(n^2)$.  We now describe how to expand a  contracted cycle $W$.
 
Let $G^c = (N^c, A^c)$ denote the contracted network.   Let $W$ be the cycle that was most recently contracted to obtain the network $G^c$.   We will show how to expand the optimal flow for $G^c$.    A node in $N^c$ that was obtained by contracting a cycle is referred to as a {\it merged node}.  The other nodes of $N^c$ were also in $N$.  They are referred to as {\it original nodes}.

Let $x$ denote the preflow in the network $G^c$ immediately prior to its contraction of the abundant cycle $W$.  At the time that $W$ was contracted, $r_{ij}(x) \ge M\Delta$ for each $(i, j) \in W$.    Let $y$ denote the optimal flow in the contracted network immediately prior to expanding $W$.   Thus $y$ is not defined for arcs of $W$.   Let $(S, T)$ denote the minimum cut in the contracted network.

The minimum cut in the  network obtained by expanding $W$ is $(S', T')$ obtained as follows.   If $w \in S$, then $S' = (S\setminus w) \cup W$, and $T' = T$.  Otherwise, $T' = (T\setminus w) \cup W$, and $S' = S $.  The capacity of $(S, T)$ is equal to the capacity of $(S', T')$.

%

We now want to modify $y$ to obtain a maximum $s$-$t$ flow.   It suffices to modify $y$ on arcs of $W$ so that the resulting flow is feasible.   We first extend $y$ so that it is defined on arcs of $W$ as follows.   For each arc $(i, j) \in A$ with $i \in W$ and $j \in W$, we let  
 $y_{ij} = x_{ij}$.
 By Lemma \ref{lem:totalpushes}, 


\begin{equation}
\sum_{i \in W} |e_y(i)| < 5n^2\Delta. 
\end{equation}

We next show how to obtain a flow $y'$ on the arcs of $W$ is \textit{feasible} so that $y + y'$ is feasible for the network obtained from $G^c$ by expanding $W$. That is, for all $v \in W$, $e_{y+y'}(v)= 0$, and $y' \ge 0$.  (We ignore arc capacities for now.)
Choose a vector $y'$ defined on the arcs of $W$ that minimizes $\sum_{(i, j) \in W} |y'_{ij}|$ and so that $y+y'$ is feasible as restricted to nodes of $W$. 
 Then  $y + y'$ is feasible for $N^c$.    Moreover, for each $(i, j) \in W$, 
 $0 \le y'_{ij} < M \Delta$. 

The time to compute $y'$ is $O(|W|^2)$ since one can restrict attention to spanning tree flows (basic flows), and there are only $O(|W|)$ spanning subtrees of $W$.   The time to expand a single cycle in this manner is $O(|W|^2)$. The time to expand all cycles is $O(n^2)$ since there are $O(n)$ arcs in all of the cycles that were contracted.  We note that the running time for expanding cycles can be improved further, but we do not do so here since it is not a bottleneck of our algorithm.

\section{Evaluating $\hat e(v, \Delta)$ and IMB$(v, x, \Delta)$ in $O(1)$  time}
\label{sec:biabundant}

In order to evaluate $\hat e(v, \Delta)$, we need to  identify whether a node is normal or special. To identify whether a node is normal or special, the algorithm maintains  IMB$(v, x, \Delta)$ for each node $v$ and for each preflow $x$.  IMB$(v, x, \Delta)$ is updated at the following times:  (1) when there is a push in an arc, (2) when an arc becomes anti-abundant, and (3) when an arc stops being anti-abundant and becomes bi-abundant.   In addition, to determine whether a node is special, we need to determine whether the node is incident to a medium arc.   For this, we need to know (4) when an arc transitions from small to medium, and (5) when an arc transitions from medium to large.  (We also update data structures following a contraction.  This is easily accomplished in the $O(m)$ time that the algorithm spends on each contraction.) 

Updating IMB after a push is straightforward.   We next deal with (2), (4) and (5).  Recall that for an arc to become anti-abundant, the arc must first become large.    For each 
node $i$, we store nodes in $A^+(i)$ in non-
increasing order of bi-capacity.   At each scaling phase, we maintain 
a pointer to the first medium arc of $A^+(i)$ as well to the first small 
arc of $A^+(i)$.   For each node $i$ the total time to update these two pointers over all 
scaling phases is $O(|A^+(i)| + \text{\# of scaling phases})$.  When we update the pointers, we can determine when an arc becomes medium and when an arc becomes large.   

Finally, we deal with (3).   There are two ways that an arc $(i,j)$ can transition from  anti-abundant to bi-abundant.   Suppose that the algorithm is in the $\Delta$-scaling phase for some $\Delta \in \Psi$.   For each subsequent scaling phase $\Delta' \in \Psi$, we let Bi-Abundant$(\Delta')$ denote the ``bucket'' of arcs that will become bi-abundant at the $\Delta'$-scaling phase, assuming that there is no further change in flow in the arc.  We store the buckets for the next $2Q+1$ scaling phases, which is sufficient.  Whenever there is a push in a large arc $(i, j)$, we check whether $(i, j) $ has become bi-abundant.   If not, we update the bucket containing $(i, j)$.   When an arc changes from one bucket to another it shifts at most 2 buckets to the left or right.  (It is also possible for an arc with residual capacity 0 to gain residual capacity.   In this case, the arc is stored in the last or second to last bucket; (the last bucket contains arcs that will be abundant after $2Q+1$ additional phases.)  In any case, the updating of buckets following a push takes $O(1)$ steps.  Moreover, the updating of buckets does not require the operation of division nor the operation of taking logs.

%
%
%
%
%

We summarize our results in the following lemma.

\begin{lem}
\label{lem:evalehat}   
At each iteration of the algorithm, for each node $v \in N(\Delta)$ and for each preflow x, $IMB(x, v, \Delta) $ and
$\hat e(v, \Delta)$ can each be evaluated in $O(1)$ time.  
\end{lem}

\end{appendix}



\end{document}